\documentclass[journal,10pt]{IEEEtran}

\usepackage{amsmath,amssymb,amsfonts}
\usepackage{bm} 
\usepackage{amsthm}
\usepackage{bbm, dsfont}
\usepackage{mathtools}
\usepackage{cite} 
\usepackage{graphics} 
\usepackage{float}
\usepackage{booktabs}
\usepackage{epstopdf}
\usepackage{url} 
\usepackage{paralist}
\usepackage{lipsum}
\usepackage{mathtools}
\usepackage{breqn}
\usepackage{cuted}
\usepackage{graphicx}
\usepackage[caption=false,font=footnotesize]{subfig}
\usepackage{optidef}
\usepackage{algorithm,algorithmic}
\usepackage{xcolor}
\usepackage[framemethod=TikZ]{mdframed}
\usepackage{balance}
\usepackage{comment}
\usepackage{cleveref}

\allowdisplaybreaks 

\newtheorem{lem}{Lemma}

\newcounter{tempEquationCounter}
\newcounter{thisEquationNumber}
\newenvironment{floatEq}
{\setcounter{thisEquationNumber}{\value{equation}}\addtocounter{equation}{1}
\begin{figure*}[!t]
\normalsize\setcounter{tempEquationCounter}{\value{equation}}
\setcounter{equation}{\value{thisEquationNumber}}
}
{\setcounter{equation}{\value{tempEquationCounter}}
\hrulefill\vspace*{4pt}
\end{figure*}}

\begin{document}
%
\title{Fluid-Antenna-Enabled Integrated Bistatic Sensing and Backscatter Communication Systems}

\author{A. Abdelaziz Salem, 
        Saeed Abdallah*,~\IEEEmembership{Member,~IEEE},
        Khawla Alnajjar,~\IEEEmembership{Member,~IEEE}, 
        Mahmoud A. Albreem,~\IEEEmembership{Senior Member,~IEEE},
        Mohamed Saad,~\IEEEmembership{Senior Member,~IEEE}, 
        Hayssam Dahrouj, ~\IEEEmembership{Senior Member,~IEEE}
        Hesham Elsawy, ~\IEEEmembership{Senior Member,~IEEE}
        
\thanks{A. Abdelaziz Salem is with the Smart Automation and Communication Technologies (SACT) Research Center,  Research Institute of Sciences and Engineering (RISE), University of Sharjah, PO Box 26666, Sharjah, UAE. He is also with the Department of Electronics and Electrical Communications Engineering, Faculty of Electronic Engineering, Menoufia University, Menouf 32952, Menoufia, Egypt, e-mail: \texttt{ahmed.abdalaziz40@el-eng.menofia.edu.eg}.}
\thanks{\noindent Saeed Abdallah, Khawla Alnajjar and Mahmoud A. Albreem are with the Department of Electrical Engineering, University of Sharjah, Sharjah, UAE, e-mails: \texttt{\{sabdallah,kalnajjar,malbreem\}@sharjah.ac.ae}.}
\thanks{\noindent Mohamed Saad and Hayssam Dahrouj are with the Department of Computer Engineering, University of Sharjah, Sharjah, UAE, e-mails: \texttt{msaad@sharjah.ac.ae,hayssam.dahrouj@gmail.com}.}
\thanks{\noindent Hesham Elsawy is with the School of Computing, Queen's University, Kingston, Ontario, Canada, e-mail: \texttt{hesham.elsawy@queensu.ca}.}
\thanks{$^*$ Corresponding author.}
\thanks{This research was supported in part by the Smart Automation and Communication Technologies (SACT) Research Center,
Research Institute for Science and Engineering (RISE), University of Sharjah, Operating Grant no. 150410 and Competitive
Grant no. 25020403357.}
}

{}

\maketitle

\begin{abstract}
This paper studies a fluid-antenna-enabled integrated bistatic sensing and backscatter communication system for future networks where connectivity, power delivery, and environmental awareness are jointly supported by the same infrastructure. A multi-antenna base station (BS) with transmitting fluid antennas serves downlink users, energizes passive tags, and illuminates radar targets, while a spatially separated multi-antenna reader decodes tag backscatter and processes radar echoes to avoid the strong self-interference that would otherwise obscure weak returns at the BS. The coexistence of tags and targets, however, induces severe near--far disparities and multi-signal interference, which can be mitigated by fluid antennas through additional spatial degrees of freedom that reshape the multi-hop channels. We formulate a transmit-power minimization problem that jointly optimizes the BS information beamformers, sensing covariance matrix, reader receive beamformers, tag reflection coefficients, and fluid-antenna (FA) positions under heterogeneous quality of service constraints for communication, backscatter, and sensing, as well as energy-harvesting and FA geometry requirements. To tackle the resulting non-convex problem, we develop an alternating-optimization block-coordinate framework that solves four tractable subproblems using semidefinite relaxation, majorization--minimization, and successive convex approximation. Numerical results show consistent transmit-power savings over fixed-position antennas and zero-forcing baselines, achieving about 13.7\% and 54.5\% reductions, respectively.
\end{abstract}


\IEEEpeerreviewmaketitle

\section{Introduction}\label{sec1}
Integrated sensing and communication (ISAC) is emerging as a key enabler for spectrum- and hardware-efficient networks, where a single downlink transmission simultaneously supports multi-user communications and environmental sensing. In parallel, backscatter communication (BComs) has attracted significant attention as a promising technology for ultra-low-power Internet-of-Things (IoT) devices, where passive tags modulate and reflect incident carrier signals \cite{rezaei2023coding}. Combining ISAC and BComs in a unified framework, referred to as integrated sensing and backscatter communication (ISABC), promises dense connectivity and high-resolution sensing at low cost and energy expenditure \cite{galappaththige2023integrated}, making it well-suited to vehicular networks, smart warehousing, and industrial automation.

From a conceptual standpoint, ISABC can be viewed as a device-centric ISAC variant with active sensing \cite{galappaththige2023integrated}. Since tags can carry embedded data while imprinting environment-dependent information onto their reflections, a reader can exploit tag-reflected signals for both communication and sensing, potentially improving spectral and hardware efficiency. Nevertheless, the resulting composite received signal typically contains a mixture of user data, tag backscatter, and clutter reflections, which motivates advanced decoding architectures \cite{galappaththige2025optimization}.

Recently, fluid-antennas (FAs) have attracted considerable research attention due to their ability to dynamically adapt their position and orientation \cite{ning2025movable}. This physical reconfigurability enables FAs to reshape the effective propagation channel by adjusting the antenna’s spatial placement and pointing direction, thereby improving signal transmission and reception. FAs have been advocated as a means to address inherent limitations of conventional multiple-input multiple-output (MIMO) systems, where fixed antenna locations restrict the ability to exploit spatial channel variations within a given transmit/receive region \cite{zhu2023movable}. In ISABC, the base station (BS) simultaneously supports downlink transmission, tag energization/backscattering, and target illumination, and so the resulting channels become sensitive to geometry and subject to severe near--far effects and strong coupling between the communication and sensing links. Consequently, the  spatial reconfigurability of FAs becomes useful for ISABC by providing additional spatial degrees of freedom (DoFs) to shape the composite channels by repositioning the radiating element(s) within a prescribed aperture \cite{fang2025integrated}. 
Despite these advantages, FA-enabled ISABC deployments should address interference in dense IoT settings, stringent tag energy budgets, and hardware limitations. Accordingly, ISABC architectures should be designed with low-complexity processing in mind and should accommodate devices powered by highly-constrained sources or energy harvesting (EH) \cite{3gpp_TR38848_2023}.

Most existing ISABC studies focus on scenarios where backscatter tags act as the primary sensing objects, thereby overlooking the coexistence of energy-constrained tags and independent radar targets. This paper generalizes ISABC by modeling both within a unified bistatic architecture, where a dedicated reader collects tag backscatter and target echoes. Moreover, we leverage an FA-enabled BS to introduce additional spatial DoFs for beam steering and diversity, achieving a balance between tag decoding reliability, downlink communication QoS, and target tracking performance. By reconfiguring the transmit aperture toward favorable spatial locations, the proposed approach can strengthen inherently weak tag-reflected links, mitigate direct BS--reader interference, and reduce cross-interference between active-user signals and tag modulation while considering the EH constraint.

\subsection{Related works}
As discussed earlier, DoFs are fundamental to achieving both high-throughput communication and high-accuracy sensing in ISAC systems. Accordingly, a variety of advanced technologies have been integrated with ISAC to enrich the available spatial and waveform DoFs, including multiple-input multiple-output (MIMO) and massive MIMO architectures, antenna selection strategies, and reconfigurable intelligent surfaces (RIS) \cite{li2024framework, topal2024multi, liu2023joint, jiang2024exploiting}. For example, MIMO-based designs have been shown to provide sufficient spatial flexibility for joint beamforming in dual-functional radar--communication systems, particularly at mmWave frequencies \cite{li2024framework}. Massive MIMO further expands the spatial DoFs and, when combined with mechanisms such as beam tilting and beam splitting, can improve the sensing--communication tradeoff and enable high-rate operation in mmWave/terahertz bands \cite{topal2024multi}. In addition, antenna selection strategies, often jointly optimized with transmit beamforming, have been investigated to balance multi-user communication requirements against radar sensing objectives under practical hardware and resource constraints \cite{liu2023joint}. Furthermore, RIS-assisted ISAC architectures have been proposed for cluttered or obstructed propagation environments, where controllable reflections can reshape channel conditions and provide additional DoFs to enhance both sensing and communication performance \cite{jiang2024exploiting}.

The integration of backscatter communication with ISAC has recently attracted growing interest, motivated by the potential for ultra-low-power connectivity and sensing. An early study has introduced an ISABC system and developed a power-allocation strategy to jointly enhance communication and sensing performance \cite{galappaththige2023integrated}, while another work has provided a broad overview of this integration, outlining key design principles and open challenges \cite{zargari2023sensing}. Protocol-level designs have also been explored, including backscatter-based ISAC frameworks with dedicated processing procedures and algorithms for joint localization and data transmission \cite{huang2022integrated}, as well as ambient backscatter-aided ISAC schemes with tailored estimation algorithms for V2X networks \cite{li2023ambc}. 

Moreover, rate-splitting multiple access has been exploited to mitigate interference of ISABC via optimizing transmit beamformers, user common rates, and tag reflection coefficients \cite{galappaththige2025optimization}. Further contributions have examined symbiotic architectures that unify localization and ambient backscatter, enabling concurrent communication and sensing functionalities \cite{ren2023toward}, and have exploited radar clutter as an effective carrier to support ambient backscatter links \cite{venturino2023radar}. In addition, advanced receiver and beamforming designs have been reported, such as two-dimensional direction-of-arrival sensing with multi-tag symbol detection \cite{tao2024integrated} and joint transmit beamforming strategies that minimize transmit power while satisfying both tag detection and communication requirements in ISABC systems \cite{luo2024isac}.

FA technology is also investigated as a practical means to reshape wireless channels, improve link reliability and capacity, and reduce energy consumption in challenging propagation environments. Existing studies establish physics- and field-based channel models for FAs by characterizing how the received response varies with antenna position using per-path amplitude/phase and angular information under far-field assumptions \cite{zhu2023modeling}. Building on such modeling, FA-enabled MIMO architectures demonstrate notable capacity gains over conventional fixed-position antenna (FPA) systems by jointly optimizing transmit and/or receive antenna locations \cite{ma2023mimo}. Beyond throughput improvements, FAs are further exploited to enhance physical-layer security by using position-induced channel variations to strengthen legitimate links and degrade eavesdropping channels \cite{tang2024secure}. In multi-user settings, FA-assisted multiple-access designs are formulated via power-minimization frameworks that jointly optimize antenna positions, user transmit powers, and BS combining to quantify FA gains under practical constraints \cite{zhu2023movable}. Related investigations in interference-limited MISO channels demonstrate that FA-provided spatial degrees of freedom can simultaneously boost the desired signal strength and suppress interference, yielding significant transmit-power reductions \cite{wang2024movable}. Moreover, \cite{salem2025movable} studies an FA-assisted covert ISAC framework with non-orthogonal multiple access, where a multi-antenna BS jointly serves public and covert user pairs through shared communication-and-sensing beams, using FAs to enhance covert communications without sacrificing service reliability.

Despite the above progress, existing works remain limited in capturing practical ISABC deployments. First, many ISABC formulations implicitly treat backscatter tags as substitutes for conventional radar targets. This assumption is crucial, as tags intentionally modulate extremely weak backscatter signals under stringent EH and operational constraints, while radar targets are uncontrolled objects with diverse radar cross-sections (RCS) and independent motion characteristics. The coexistence of these two classes of reflectors introduces practical challenges, including severe near–-far power disparities and the need to balance waveform and beamforming resources between reliable tag decoding and accurate target detection. Specifically, the near–-far disparity is particularly pronounced because the tag path is inherently double-hop \cite{galappaththige2025optimization}. The incident waveform must first propagate to the tag, and only a very small fraction of the received energy is then backscattered to the receiver. In addition, radar targets typically exhibit RCS that can be tens of decibels larger than the effective RCS of a passive tag \cite{borgese2020radar}. Consequently, returns from a nearby tag or a large object (e.g., a vehicle) can be tens, often exceeding 100 dB, stronger than those from a distant tag or a small target, leading to stringent receiver dynamic-range requirements and causing signal masking and mutual interference \cite{kumari2017ieee}.

Furthermore, a large portion of the literature adopts monostatic or co-located transceiver architectures, which can suffer from severe self-interference and dynamic coupling \cite{hakimi2022sum} when weak backscatter and radar echoes are received at the same node that transmits high-power waveforms \cite{jiao2025information}. In addition, although FA techniques have demonstrated substantial gains in communication and ISAC systems \cite{wong2020fluid, lyu2025movable}, their role in integrated bistatic ISABC systems and their interaction with joint tag decoding and target sensing objectives have not been thoroughly investigated. These gaps motivate a unified framework that explicitly models tag--target coexistence and leverages a bistatic reader together with a FA-enabled BS to provide additional spatial DoFs for balancing reliable tag decoding and accurate target sensing.

\subsection{Contribution and paper organization}
Motivated by smart-city, vehicular, warehousing, and industrial IoT deployments, a single wireless infrastructure is increasingly expected to (i) deliver downlink data to active users, (ii) energize and decode passive backscatter tags, and (iii) sense and track environmental targets. In such settings, tag-reflected links are typically orders of magnitude weaker than user links and can be overwhelmed by clutter and strong target echoes, while monostatic designs further face stringent isolation requirements due to simultaneous transmission and reception. To address these challenges, we consider a bistatic ISABC architecture in which a spatially separated reader collects weak reflections while mitigating self-interference. Moreover, by equipping the BS with a FA, additional reconfigurable aperture DoFs become available to strengthen tag backscatter, suppress dominant target components, and reduce multi-signal interference, thereby enabling a more effective tradeoff between communication service, tag reliability, and sensing accuracy. The main contributions are summarized as follows.
\begin{itemize}
    \item We propose an FA-enabled bistatic ISABC framework where an FA-equipped BS serves downlink users, energizes passive tags, and illuminates radar targets, while a separated reader jointly decodes tag backscatter and processes target echoes. Unlike conventional ISABC formulations, the model explicitly captures the coexistence of energy-constrained, intentionally modulated tags and uncontrolled radar targets, including near--far disparities and multi-signal interference.
    \item We formulate a total BS transmit-power minimization problem that jointly designs downlink precoders, sensing covariance, reader receive beamformers, tag reflection coefficients, and the FA position, subject to heterogeneous QoS constraints across communication, backscatter, bistatic sensing, and EH. 
    \item To solve the resulting nonconvex complex optimization problem, we  develop an alternating optimization (AO)-based block-coordinate algorithm that decomposes the main problem into tractable subproblems: reader and BS beamforming updates are handled via semi-definite relaxation (SDR), while FA positioning is addressed via majorization–minimization (MM) and successive convex approximation (SCA).
    \item Numerical results verify consistent transmit-power reductions over established baselines across a wide range of parameters. For instance, at $16$ FA antennas, the proposed design achieves about $13.7\%$ and $54.5\%$ power savings relative to fixed-position antennas (FPA) and zero-forcing (ZF), respectively.
\end{itemize}

The remainder of this paper is organized as follows. Section~\ref{sec2} describes the proposed system model. Section~\ref{sec3} presents the problem formulation and solution. Section~\ref{sec4} reports the simulation results and performance evaluation. Finally, Section~\ref{sec5} concludes the paper.

\textit{Notation:} Vectors are denoted by lowercase boldface letters, matrices by uppercase boldface letters, and scalars by regular letters. The notation $\mathbb{C}^{M \times N}$ denotes the space of $M \times N$ complex-valued matrices. We use $\mathrm{Tr}(\cdot)$ for the trace operator and  $\mathrm{diag}(\cdot)$ to form a diagonal matrix of a given vector. The symbols $(\cdot)^*$ and $(\cdot)^H$ denote the complex conjugate and the conjugate transpose, respectively, while $\Re\{\cdot\}$ and $\Im\{\cdot\}$ represent the real and imaginary parts. The $M \times M$ identity matrix is written as $\mathbf{I}_M$. Furthermore, $\mathbf{x} \sim \mathcal{CN}(\boldsymbol{\mu},\mathbf{R})$ denotes a circularly symmetric complex Gaussian random vector with mean $\boldsymbol{\mu}$ and covariance matrix $\mathbf{R}$.


\section{System model}\label{sec2}
\begin{figure}[t]
	\centering{\includegraphics[width=0.9\columnwidth]{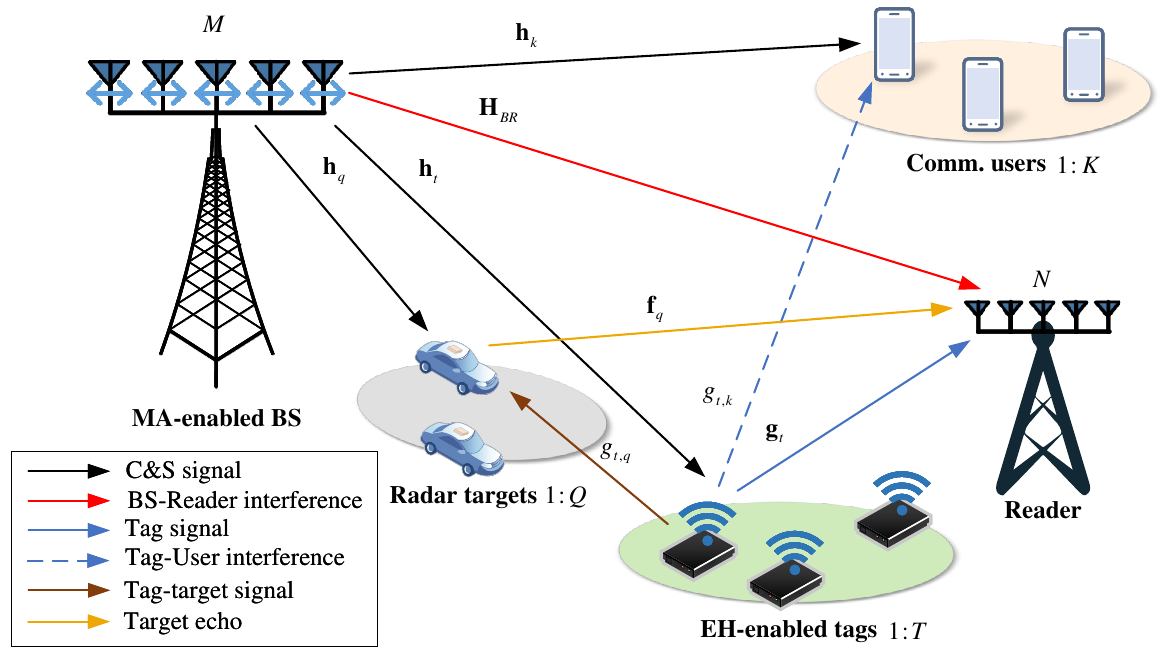}}
    \captionsetup{skip=3pt}
	\caption{FA-enabled bistatic ISABC system.\label{sys_model}}
\end{figure}
We consider an FA-enabled bistatic ISABC setup, where a BS is equipped with $M$ transmitting FAs. The BS serves $K$ single antenna users (${{U}_{k}}$ denotes the $k$-th user) and $T$ single antenna tags (${{A}_{t}}$ denotes the $t$-th tag), while a separate reader, with $N$ receiving antennas, is employed to capture the backscattered communications and sense $Q$ radar targets (${{T}_{q}}$ denotes the $q$-th target). 
Let $\mathcal{U}=\{1,\dots,K\}$, $\mathcal{M}=\{1,\dots,M\}$, $\mathcal{N}=\{1,\dots,N\}$,
$\mathcal{A}=\{1,\dots,T\}$, and $\mathcal{T}=\{1,\dots,Q\}$ denote the index sets of users, transmitting FAs, receiving antennas, tags, and targets, respectively. With a separate reader, we enable bistatic sensing and backscatter communication while avoiding  self-interference at the BS. An illustrative example of the considered model is shown in Fig.~\ref{sys_model}. The position vector of the FA of the BS is given by  $\mathbf{z}={{\left[ {{z}_{1}},...{{z}_{M}} \right]}^{T}}\in {{\mathbb{R}}^{M\times 1}}$, where ${{z}_{m}}$ denotes the position of $m$-th element. Without loss of generality, the FA repositioning overhead is considered negligible compared to the channel coherence time under the assumption of narrow-band quasi-static channels \cite{ma2023mimo}. 

Given the negligible FA displacement relative to the link distance, the far-field plane-wave model \cite{zhu2023movable} applies. Under this model, the angle of departure (AoD), the angle of arrival (AoA), and the amplitude coefficients remain constant, while only the path phase varies with FA position \cite{zhu2023modeling}. Accordingly, we adopt the field-response-based model of \cite{zhu2023modeling}, representing the channel as a superposition of multi-path signal coefficients. Moreover, we assume a block-fading and line-of-sight (LoS)-dominant Rician channel model as in \cite{you2020wireless} for all the communication links, where  ${{\mathbf{h}}_{k}}$ denotes BS-$U_k$ communication channels and  $\mathbf{h}_t,\mathbf{h}_q\in\mathbb{C}^{M\times 1}$ define the BS--$A_t$ and BS--$T_q$ channels, respectively. Moreover, $\mathbf{g}_t,\mathbf{f}_q\in\mathbb{C}^{N\times 1}$ are the $A_t$--reader and $T_q$--reader channels. The direct interference channel from the BS to the reader is ${{\mathbf{H}}_{BR}}\in {{\mathbb{C}}^{M\times N}}$. In addition, ${{g}_{t,k}}$, and ${{g}_{t,q}}$ represent the ${{A}_{t}}$-${{U}_{k}}$, and ${{A}_{t}}$-${{T}_{q}}$ channel links,  respectively.

    \subsection{Communication, BComs, and sensing models}
    The BS transmits a signal $\mathbf{x}=\sum\limits_{k\in \mathcal{U}}{{{\mathbf{w}}_{k}}{{{\bar{x}}}_{k}}}+\mathbf{s}$ to enable ISABC services, where $\left\{ {{{\bar{x}}}_{k}} \right\}$ and $\mathbf{s}$ are mutually independent. Specifically, ${{\mathbf{w}}_{k}}\in {{\mathbb{C}}^{M\times 1}}, \text{ } {{\bar{x}}_{k}}$, and $\mathbf{s}\in {{\mathbb{C}}^{M\times 1}}$ represent the precoding vector for ${{U}_{k}}$, the unit-power  symbol for ${{U}_{k}}$, and the sensing signal. Then, the covariance matrix of the transmitted signal is then given by ${{\mathbf{R}}_{x}}=\mathbb{E}\left\{ \mathbf{x}{{\mathbf{x}}^{H}} \right\}=\sum\limits_{k=1}^{K}{{{\mathbf{w}}_{k}}\mathbf{w}_{k}^{H}}+{{\mathbf{R}}_{s}}$, where ${{\mathbf{R}}_{s}}=\mathbb{E}\left\{ \mathbf{s}{{\mathbf{s}}^{H}} \right\}$. To enable BComs, each tag employs a $\bar{M}$-ary modulation scheme~\cite{rezaei2023coding}, where $\bar{M}$ denotes the modulation order. Specifically, the backscatter data symbol $d_t$ is drawn from a $\bar{M}$-ary constellation and is normalized to have unit average power, i.e., $\mathbb{E}\{|d_t|^2\}=1$. As a result, the users receive both the direct signal from the BS and the backscattered signals from the tags. Hence, the signal received at ${{U}_{k}}$ is given by
    \begin{equation} \label{RX_signal_UE}
        {{y}_{k}}=\mathbf{h}_{k}^{H}\mathbf{x}+\sum\limits_{t=1}^{T}{\sqrt{{{\beta }_{t}}}\mathbf{h}_{k,t}^{H}\mathbf{x}{{d}_{t}}}+{{n}_{k}},
    \end{equation}
    where ${{\mathbf{h}}_{k,t}}={{\mathbf{h}}_{t}}{{g}_{t,k}}$ represents the cascaded backscatter link (BS- ${{A}_{t}}$ -${{U}_{k}}$), ${{\beta }_{t}}\in \left( 0,1 \right)$ is the backscattering reflection coefficient, and ${{n}_{k}}\sim\mathcal{C}\mathcal{N}\left( 0,\sigma _{k}^{2} \right)$ represents the additive white Gaussian noise (AWGN) at ${{U}_{k}}$. Then, the received signal-to-interference-plus-noise ratio (SINR) at ${{U}_{k}}$ is represented by\footnote{At the user, target echoes are typically far below the direct downlink and multi-path components. This leads to extreme direct-to-target disparities that exceed practical ADC/front-end dynamic range \cite{berger2010signal}.} 
    \begin{equation} \label{SINR_UE}
        {{\gamma }_{k}}=\frac{\mathcal{S}_k}{\mathcal{I}_k},
    \end{equation}
    where $\mathcal{S}_k = {{\left| \mathbf{h}_{k}^{H}{{\mathbf{w}}_{k}} \right|}^{2}}$, 
    $\mathcal{I}_k = \sum\limits_{i=1,i\ne k}^{K}{{{\left| \mathbf{h}_{k}^{H}{{\mathbf{w}}_{i}} \right|}^{2}}}+\mathrm{Tr}\left( \bm{\overset{\scriptscriptstyle\frown}{\mathbf H}}_{k}{{\mathbf{R}}_{s}} \right)+\sum\limits_{t=1}^{T}{{{\beta }_{t}}\left( \sum\limits_{j=1}^{K}{{{\left| \mathbf{h}_{k,t}^{H}{{\mathbf{w}}_{j}} \right|}^{2}}+\mathrm{Tr}\left( \bm{\overset{\scriptscriptstyle\frown}{\mathbf H}}_{k,t}{{\mathbf{R}}_{s}} \right)} \right)}+\sigma _{k}^{2}$, with $\bm{\overset{\scriptscriptstyle\frown}{\mathbf H}}_{k}={{\mathbf{h}}_{k}}\mathbf{h}_{k}^{H}$ and $\bm{\overset{\scriptscriptstyle\frown}{\mathbf H}}_{k,t}={{\mathbf{h}}_{k,t}}\mathbf{h}_{k,t}^{H}$. 
    
    The tag’s backscattered signals and the echo signal from the radar targets propagate back to the reader. Consequently, the reader can extract environmental information from these reflections and decode the BComs signals. Then, the received signal at the reader, ${{\mathbf{y}}_{R}}\in {{\mathbb{C}}^{N\times 1}}$, is expressed as
    \begin{multline} \label{RX_reader}
        {{\mathbf{y}}_{R}}=\underbrace{\sum\limits_{t=1}^{T}{\sqrt{{{\beta }_{t}}}\mathbf{G}_{t}^{H}\mathbf{x}{{d}_{t}}}}_{\text{BComs signals}} + \underbrace{\sum\limits_{q=1}^{Q}{{{\mu }_{q}}\mathbf{F}_{q}^{H}\mathbf{x}}}_{\text{Target echos}} \\
        + \underbrace{\sum\limits_{q=1}^{Q}{\sum\limits_{t=1}^{T}{{{\mu }_{q}}\sqrt{{{\beta }_{t}}}\mathbf{F}_{t,q}^{H}\mathbf{x}{{d}_{t}}}}}_{\text{BComs-target echos}} 
        +  \underbrace{\mathbf{H}_{BR}^{H}\mathbf{x}}_{\text{BS-reader interference}} + {{\mathbf{n}}_{R}},
    \end{multline}
    where ${{\mathbf{G}}_{t}}={{\mathbf{h}}_{t}}\mathbf{g}_{t}^{H}\in {{\mathbb{C}}^{M\times N}}$, ${{\mathbf{F}}_{q}}={{\mathbf{h}}_{q}}\mathbf{f}_{q}^{H}\in {{\mathbb{C}}^{M\times N}}$, ${{\mathbf{F}}_{t,q}}=g_{t,q}^{*}{{\mathbf{h}}_{t}}\mathbf{f}_{q}^{H}\in {{\mathbb{C}}^{M\times N}}$ represent the equivalent BS-${{A}_{t}}$-reader link, BS-${{T}_{q}}$-reader link, and BS-${{A}_{t}}$-${{T}_{q}}$-reader link, respectively. Moreover, ${{\mathbf{n}}_{R}}\sim\mathcal{C}\mathcal{N}\left( 0,\sigma _{R}^{2}{{\mathbf{I}}_{N}} \right)$ denotes the AWGN at the reader and ${{\mu }_{q}}\sim\mathcal{C}\mathcal{N}\left( 0,{{\upsilon }^{2}} \right)$ denotes the radar cross section (RCS) of ${{T}_{q}}$. It is worth noting that the third term in \eqref{RX_reader} can be leveraged for tag-enabled cooperative sensing, where the tag-reflected signals are used to sense the $Q$ radar targets. To extract each tag’s information while simultaneously sensing the radar targets, we deploy a set of linear combiners $\{ \mathbf{u}_{t}, \mathbf{v}_{q}\} \subset \mathbb{C}^{N\times 1}$ to decode the $t$-th tag signal and the $q$-th target echo, respectively. Therefore, the post-processed SINR at the reader for decoding the tag signal (BComs mode) or retrieving the sensing information of target (sensing mode) is expressed as 
        \begin{equation}\label{post_processed_SINR_generic}
            \gamma_{R,l}^{\bar{o}} \triangleq
            \frac{\mathcal{S}_{R,l}^{\bar{o}}}{\mathcal{I}_{R,l}^{\bar{o}}},
            \qquad
            \begin{cases}
            \bar{o}=\mathrm{BComs},\; l=t\in\mathcal{A},\\
            \bar{o}=\mathrm{sens},\; l=q\in\mathcal{T},
            \end{cases}
        \end{equation}        
    where the corresponding intended and interference-plus-noise signals are provided in (\ref{BComs_and_Sensing_SINRs_terms}), at the top of next page.
    \begin{floatEq}
	       \begin{subequations}\label{BComs_and_Sensing_SINRs_terms}\begin{align} 
		        \mathcal{S}_{R,t}^{\mathrm{BComs}}=&{{\beta }_{t}}\left( \mathbf{u}_{t}^{H}\mathbf{G}_{t}^{H}{{\mathbf{R}}_{s}}{{\mathbf{G}}_{t}}{{\mathbf{u}}_{t}}+\sum\limits_{j=1}^{K}{{{\left| \mathbf{u}_{t}^{H}\mathbf{G}_{t}^{H}{{\mathbf{w}}_{j}} \right|}^{2}}} \right), \label{Bcoms_useful} \\ 
                \mathcal{I}_{R,t}^{\mathrm{BComs}}=&\sum\limits_{i=1,i\ne t}^{T}{{{\beta }_{i}}\left( \mathbf{u}_{t}^{H}\mathbf{G}_{i}^{H}{{\mathbf{R}}_{s}}{{\mathbf{G}}_{i}}{{\mathbf{u}}_{t}}+\sum\limits_{j=1}^{K}{{{\left| \mathbf{u}_{t}^{H}\mathbf{G}_{i}^{H}{{\mathbf{w}}_{j}} \right|}^{2}}} \right)}+{{\upsilon }^{2}}\sum\limits_{q=1}^{Q}{\left( \mathbf{u}_{t}^{H}\mathbf{F}_{q}^{H}{{\mathbf{R}}_{s}}{{\mathbf{F}}_{q}}{{\mathbf{u}}_{t}}+\sum\limits_{j=1}^{K}{{{\left| \mathbf{u}_{t}^{H}\mathbf{F}_{q}^{H}{{\mathbf{w}}_{j}} \right|}^{2}}} \right)} \nonumber \\ 
                &+{{\upsilon }^{2}}\sum\limits_{q=1}^{Q}{\sum\limits_{i=1}^{T}{{{\beta }_{i}}\left( \mathbf{u}_{t}^{H}\mathbf{F}_{i,q}^{H}{{\mathbf{R}}_{s}}{{\mathbf{F}}_{i,q}}{{\mathbf{u}}_{t}}+\sum\limits_{j=1}^{K}{{{\left| \mathbf{u}_{t}^{H}\mathbf{F}_{i,q}^{H}{{\mathbf{w}}_{j}} \right|}^{2}}} \right)}}+\mathbf{u}_{t}^{H}\mathbf{H}_{BR}^{H}{{\mathbf{R}}_{s}}{{\mathbf{H}}_{BR}}{{\mathbf{u}}_{t}}+\sum\limits_{j=1}^{K}{{{\left| \mathbf{u}_{t}^{H}\mathbf{H}_{BR}^{H}{{\mathbf{w}}_{j}} \right|}^{2}}}+\sigma _{R}^{2}{{\left\| {{\mathbf{u}}_{t}} \right\|}^{2}},    \label{Bcoms_interf} \\ 
                \mathcal{S}_{R,q}^{\text{sens}}=&{{\upsilon }^{2}}\left( \mathbf{v}_{q}^{H}\mathbf{F}_{q}^{H}{{\mathbf{R}}_{s}}{{\mathbf{F}}_{q}}{{\mathbf{v}}_{q}}+\sum\limits_{j=1}^{K}{{{\left| \mathbf{v}_{q}^{H}\mathbf{F}_{q}^{H}{{\mathbf{w}}_{j}} \right|}^{2}}}+\sum\limits_{t=1}^{T}{{{\beta }_{t}}\left( \mathbf{v}_{q}^{H}\mathbf{F}_{t,q}^{H}{{\mathbf{R}}_{s}}{{\mathbf{F}}_{t,q}}{{\mathbf{v}}_{q}}+\sum\limits_{j=1}^{K}{{{\left| \mathbf{v}_{q}^{H}\mathbf{F}_{t,q}^{H}{{\mathbf{w}}_{j}} \right|}^{2}}} \right)} \right), \label{sensing_useful} \\ 
                \mathcal{I}_{R,q}^{\text{sens}}=&\sum\limits_{t=1}^{T}{{{\beta }_{t}}\left( \mathbf{v}_{q}^{H}\mathbf{G}_{t}^{H}{{\mathbf{R}}_{s}}{{\mathbf{G}}_{t}}{{\mathbf{v}}_{q}}+\sum\limits_{u=1}^{K}{{{\left| \mathbf{v}_{q}^{H}\mathbf{G}_{t}^{H}{{\mathbf{w}}_{u}} \right|}^{2}}} \right)}+{{\upsilon }^{2}}\sum\limits_{j=1,j\ne q}^{Q}{\left( \mathbf{v}_{q}^{H}\mathbf{F}_{j}^{H}{{\mathbf{R}}_{s}}{{\mathbf{F}}_{j}}{{\mathbf{v}}_{q}}+\sum\limits_{i=1}^{K}{{{\left| \mathbf{v}_{q}^{H}\mathbf{F}_{j}^{H}{{\mathbf{w}}_{i}} \right|}^{2}}} \right)} \nonumber \\ 
                & +{{\upsilon }^{2}}\sum\limits_{j=1,j\ne q}^{Q}{\sum\limits_{t=1}^{T}{{{\beta }_{t}}\left( \mathbf{v}_{q}^{H}\mathbf{F}_{t,j}^{H}{{\mathbf{R}}_{s}}{{\mathbf{F}}_{t,j}}{{\mathbf{v}}_{q}}+\sum\limits_{i=1}^{K}{{{\left| \mathbf{v}_{q}^{H}\mathbf{F}_{t,j}^{H}{{\mathbf{w}}_{i}} \right|}^{2}}} \right)}}+\mathbf{v}_{q}^{H}\mathbf{H}_{BR}^{H}{{\mathbf{R}}_{s}}{{\mathbf{H}}_{BR}}{{\mathbf{v}}_{q}}+\sum\limits_{i=1}^{K}{{{\left| \mathbf{v}_{q}^{H}\mathbf{H}_{BR}^{H}{{\mathbf{w}}_{i}} \right|}^{2}}}+\sigma _{R}^{2}{{\left\| {{\mathbf{v}}_{q}} \right\|}^{2}}. \label{sensing_interf} 
	\end{align}\end{subequations}
    \end{floatEq}

    \subsection{Energy harvesting model}
    Since the tag is a passive component, it does not generate RF signals. Therefore, the tag rather employs EH to enable its essential functions. Specifically, ${{A}_{t}}$ employs a power splitter that divides the incident RF signal into two portions. Let the total received power at ${{A}_{t}}$ be $p_{t}^{in}=\sum\limits_{k=1}^{K}{{{\left| \mathbf{h}_{t}^{H}{{\mathbf{w}}_{k}} \right|}^{2}}}+\mathbf{h}_{t}^{H}{{\mathbf{R}}_{s}}{{\mathbf{h}}_{t}}$. Then, ${{A}_{t}}$ reflects a portion ${{\beta }_{t}}p_{t}^{in}$ for backscatter data transmission and harvests the remaining $\left( 1-{{\beta }_{t}} \right)p_{t}^{in}$ \cite{zhang2013mimo}. To capture real world effects such as circuit sensitivity and saturation, the non-linear model is employed as $p_{t}^{\text{EH}}=\Phi \left( \left( 1-{{\beta }_{t}} \right)p_{t}^{in} \right)$, where $\Phi \left( . \right)$ represents a nonlinear function \cite{boshkovska2015practical}. Following \cite{zargari2021max}, our design satisfies the EH requirement of $A_t$ by considering the following constraint. 
    \begin{equation} \label{EH_constraint}
        \left( 1-{{\beta }_{t}} \right)p_{t}^{in}\ge {{\Phi }^{-1}}\left( \rho  \right), ~ 0 \le {{\beta }_{t}} \le 1, ~ t\in \mathcal{A},
    \end{equation}
    where ${{\Phi }^{-1}}\left( \rho  \right) = \frac{b-\left( \rho+\frac{b}{c} \right)c}{\rho+\frac{b}{c}-a}$, where $\rho$ represents the required harvested (output/DC) power, which relies on the thresholds of saturation, sensitivity, and activation input power $a$, $b$, and $c$, respectively ($a >$0, $b>$0, and $c>$0) \cite{chen2017wireless}.

\section{Problem formulation and proposed solution} \label{sec3}
In this section, we aim to design the precoding vectors $\left\{ {{\mathbf{w}}_{k}} \right\}_{k\in \mathcal{U}}$, sensing covariance matrix ${{\mathbf{R}}_{s}}$, receive beamforming vectors $\left\{ {{\mathbf{u}}_{t}},{{\mathbf{v}}_{q}} \right\}_{t \in \mathcal{A},~ q \in \mathcal{T}}$, tag reflection coefficients $\left\{ {{\beta }_{t}} \right\}_{t \in \mathcal{A}}$, and FA  position vector $\mathbf{z}$ that minimize the total transmit power. In this context, the communication, BComs, and sensing SINRs must exceed their minimum requirements while satisfying the constraints of FA geometry and EH. Mathematically, the optimization problem can be formulated as follows:
\begin{subequations}\label{main_problem}
	\begin{align}
        \min_{{{\mathbf{w}}_{k}},{{\mathbf{R}}_{s}},{{\mathbf{u}}_{t}},{{\mathbf{v}}_{q}},\mathbf{z},{{\beta }_{t}}} \quad &
		\sum\limits_{k=1}^{K}{{{\left\| {{\mathbf{w}}_{k}} \right\|}^{2}}+\mathrm{Tr}\left( {{\mathbf{R}}_{s}} \right)}, \label{opt1_obj} \\
		\textrm{s.t.:} \quad 
        & {{\gamma }_{k}}\ge \Gamma _{k}^{\mathrm{th}}, ~ k\in \mathcal{U}, \label{opt1_c1}  \\
		& \gamma _{R,t}^{\mathrm{BComs}}\ge \Gamma _{R,t}^{\mathrm{BComs},\mathrm{th}}, ~ t\in \mathcal{A} \label{opt1_c2}, \\
  	    & \Upsilon _{R,q}^{\mathrm{sens}}\ge \Gamma _{R,q}^{\mathrm{sens},\mathrm{th}}, ~ q\in \mathcal{T}  
           \label{opt1_c3}, \\
        & {{z}_{1}}\ge 0,{{z}_{M}}\le D, \nonumber \\
        & {{z}_{m+1}}-{{z}_{m}}\ge \delta , ~ m\in \mathcal{M} \label{opt1_c4} \\
        & (\ref{EH_constraint}), \label{opt1_c5} 
	\end{align}
\end{subequations}
where $D$ and $\delta $ denote the maximum aperture of the FA array and the minimum inter-element spacing, respectively. The constraints (\ref{opt1_c1})--(\ref{opt1_c1}) denote the QoS requirements of communication, BComs, and sensing, respectively. The optimization problem in (\ref{main_problem}) is highly challenging. In addition to the non-convex communication, sensing and BComs fractional constraints, the optimization variables are jointly coupled both in the objective and in the constraints, making the solution space mathematically intractable. To tackle this, we adopt a block-coordinate approach that alternates between updating one variable block while fixing the others. Specifically, we first split the original problem (\ref{main_problem}) into four subproblems: (i) designing the transmit beamforming and the sensing covariance matrices at the BS, (ii) optimizing BComs reflection coefficients, (iii) designing the FA positions at the BS, and (iv) optimizing the receiving vectors. The details of this sequential update strategy are presented in the following subsections.

\subsection{Optimizing the transmit beamformers and the sensing covariance}\label{sec3-1}
In this section, we design the transmit beamformers ${{\left\{ {{\mathbf{w}}_{k}} \right\}}_{k\in \mathcal{U}}}$ and the covariance matrix ${{\mathbf{R}}_{s}}$ for given ${{\left\{ {{\mathbf{u}}_{t}} \right\}}_{t\in \mathcal{A}}},{{\left\{ {{\beta }_{t}} \right\}}_{t\in \mathcal{A}}},{{\left\{ {{\mathbf{v}}_{q}} \right\}}_{q\in \mathcal{T}}}$, and $\mathbf{z}$. Therefore, the optimization problem (\ref{main_problem}) becomes 
\begin{subequations}\label{BF_and_cov_optim_V1}
	\begin{align}
        \min_{{{\mathbf{w}}_{k}},{{\mathbf{R}}_{s}}} \quad &
		\sum\limits_{k=1}^{K}{{{\left\| {{\mathbf{w}}_{k}} \right\|}^{2}}+\mathrm{Tr}\left( {{\mathbf{R}}_{s}} \right)}, \label{BF_obj_V1} \\
		\textrm{s.t.:} \quad 
        & (\ref{EH_constraint}), (\ref{opt1_c1})-(\ref{opt1_c3}). \label{BF_c1_V1}  
	\end{align}
\end{subequations}
The optimization problem in (\ref{BF_and_cov_optim_V1}) is non-convex in ${{\left\{ {{\mathbf{w}}_{k}} \right\}}_{k\in \mathcal{U}}}$ and ${{\mathbf{R}}_{s}}$. To tackle this, problem in (\ref{BF_and_cov_optim_V1}) can be reformulated as a semi-definite programming (SDP) problem, where the optimization variables are defined as ${{\mathbf{W}}_{k}}={{\mathbf{w}}_{k}}\mathbf{w}_{k}^{H} \succeq \mathbf{0}, \mathbf{R}_s \succeq \mathbf{0}$  with $\mathrm{Rank}\left( {{\mathbf{W}}_{k}} \right)=1$. Hence, the communication QoS constraint (\ref{opt1_c1}) of the user $U_k$ can be expressed as 
\begin{equation}\label{comm_QoS_BF}
     \Gamma_k^{\mathrm{th}}\Big(\mathcal{I}_k+\mathcal{L}_k+\sigma_k^2\Big) \le \mathrm{Tr}\!\left(\bm{\overset{\scriptscriptstyle\frown}{\mathbf H}}_{k}\mathbf W_k\right),
\end{equation}
where the multi-user interference term and the BComs leakage term are defined as
$\mathcal{I}_k \triangleq 
\mathrm{Tr}\!\left(\bm{\overset{\scriptscriptstyle\frown}{\mathbf H}}_{k}\left(\mathbf{\bar W}_k+\mathbf R_s\right)\right)$, 
$\mathcal{L}_k \triangleq 
\mathrm{Tr}\!\left(\mathbf{\tilde H}_k \mathbf X_{\mathrm{TX}}\right)$, respectively, with $\mathbf{\tilde H}_k=\sum_{t=1}^{T}\beta_t \bm{\overset{\scriptscriptstyle\frown}{\mathbf H}}_{k,t}$,
$\mathbf X_{\mathrm{TX}}=\mathbf R_s+\mathbf W$, $\mathbf{\bar W}_k=\sum_{i=1,i\neq k}^{K}\mathbf W_i$,
and $\mathbf W=\sum_{j=1}^{K}\mathbf W_j$. 

Likewise, the BComs and sensing QoS constraints in (\ref{opt1_c2}) and (\ref{opt1_c3})
can be expressed in the following unified linear form.  
\begin{equation}\label{gen_QoS_BF}
\Gamma_{R,l}^{\bar{o},\mathrm{th}}\Big(\mathrm{Tr}\!\left(\mathbf C_{l}\mathbf X_{\mathrm{TX}}\right)
+\sigma_R^2\|\mathbf t_{l}\|^2\Big)
\le \xi_{l}\,\mathrm{Tr}\!\left(\bm{\overset{\scriptscriptstyle\frown}{\mathbf Z}}_{l}\mathbf X_{\mathrm{TX}}\right),
\end{equation}
where $\bar{o}\in \{\mathrm{BComs}, \mathrm{sens} \}$, and where 
\[
(\mathbf C_{l},\mathbf t_{l},\xi_{l},\bm{\overset{\scriptscriptstyle\frown}{\mathbf Z}}_{l})=
\begin{cases}
(\mathbf B_{t},\mathbf u_{t},\beta_{t},\bm{\overset{\scriptscriptstyle\frown}{\mathbf G}}_{t}), & l=t,\ \bar{o}=\mathrm{BComs},\\
(\mathbf A_{q},\mathbf v_{q},\upsilon^{2},\bm{\overset{\scriptscriptstyle\frown}{\mathbf F}}_{q}), & l=q,\ \bar{o}=\mathrm{sens}.
\end{cases}
\]
Moreover, the effective channels of BComs are $\bm{\overset{\scriptscriptstyle\frown}{\mathbf G}}_{t} = {{\mathbf{G}}_{t}}{{\mathbf{U}}_{t}}\mathbf{G}_{t}^{H}$, 
${{\mathbf{U}}_{t}}={{\mathbf{u}}_{t}}\mathbf{u}_{t}^{H}$, 
${{\mathbf{B}}_{t}}={{\mathbf{\bar{G}}}_{t}}+{{\upsilon }^{2}}\left( \bm{\overset{\scriptscriptstyle\frown}{\mathbf F}}_{t}+{{{\mathbf{\tilde{F}}}}_{t}} \right) + \bm{\overset{\scriptscriptstyle\frown}{\mathbf H}}_{BR,t}$, 
${{\mathbf{\bar{G}}}_{t}}=\sum\limits_{i=1,i\ne t}^{T}{{{\beta }_{i}}{{\mathbf{G}}_{i}}{{\mathbf{U}}_{t}}\mathbf{G}_{i}^{H}}$,
$\bm{\overset{\scriptscriptstyle\frown}{\mathbf F}}_{t}=\sum\limits_{q=1}^{Q}{{{\mathbf{F}}_{q}}{{\mathbf{U}}_{t}}\mathbf{F}_{q}^{H}}$,
${{\mathbf{\tilde{F}}}_{t}}=\sum\limits_{q=1}^{Q}{\sum\limits_{i=1}^{T}{{{\beta }_{i}}{{\mathbf{F}}_{i,q}}{{\mathbf{U}}_{t}}\mathbf{F}_{i,q}^{H}}}$, and 
$\bm{\overset{\scriptscriptstyle\frown}{\mathbf H}}_{BR,t}={{\mathbf{H}}_{BR}}{{\mathbf{U}}_{t}}\mathbf{H}_{BR}^{H}$. 
The effective sensing channels are
obtained from those in BComs mode by interchanging the roles of the
communication index $t$ with sensing index $q$, and by replacing
$\mathbf{U}_{t}$ with
$\mathbf{V}_{q}=\mathbf{v}_{q}\mathbf{v}_{q}^{H}$, together with the
corresponding substitutions
$\mathbf{B}_{t}\!\to\!\mathbf{A}_{q}$,
$\mathbf{\bar{G}}_{t}\!\to\!\mathbf{\bar{G}}_{q}$,
$\bm{\overset{\scriptscriptstyle\frown}{\mathbf F}}_{t},\tilde{\mathbf{F}}_{t}
\!\to\!
\mathbf{\bar{F}}_{q},\tilde{\mathbf{F}}_{q}$, and
$\bm{\overset{\scriptscriptstyle\frown}{\mathbf H}}_{BR,t}
\!\to\!
\bm{\overset{\scriptscriptstyle\frown}{\mathbf H}}_{BR,q}$.
Furthermore, the EH constraint (\ref{EH_constraint}) can be re-written as 
\begin{equation} \label{EH_QoS_BF}
    {{P}^{\mathrm{EH,th}}}\le \left( 1-{{\beta }_{t}} \right)\mathrm{Tr}\left( {{\mathbf{H}}_{t}}{{\mathbf{X}}_{\mathrm{TX}}} \right),
\end{equation}
where ${{\mathbf{H}}_{t}}={{\mathbf{h}}_{t}}\mathbf{h}_{t}^{H}$.

According to the analysis above, the optimization problem of BS transmit beamforming and covariance matrix can be reformulated as 
\begin{subequations}\label{BF_and_cov_optim_V2}
	\begin{align}
        \min_{{\mathbf{W}_{k}},{{\mathbf{R}}_{s}}} \quad &
		\mathrm{Tr}\left( \sum\limits_{k=1}^{K}  {{\mathbf{W}}_{k}} \right)+\mathrm{Tr}\left(  {{\mathbf{R}}_{s}} \right), \label{BF_obj_V2} \\
		\textrm{s.t.:} \quad 
        & (\ref{comm_QoS_BF}) - (\ref{EH_QoS_BF}), \quad \mathrm{Rank (\mathbf{W}_k)=1}. \label{BF_c1_V2}  
	\end{align}
\end{subequations}
Problem~(\ref{BF_and_cov_optim_V2}) can be tackled via SDR by dropping the rank-one constraint, which yields a convex SDP that can be efficiently solved using the CVX optimization toolbox \cite{grant2011cvx}. If the solution of  (\ref{BF_and_cov_optim_V2}) yields a rank-one matrix, the optimal transmit beamformer is directly recovered from its eigenvalue decomposition. Otherwise, Gaussian randomization \cite{wu2019intelligent} can be employed to recover rank-one transmit beamforming, while maintaining the QoS constraints.

\subsection{Optimizing the reflection coefficients for the tags}\label{sec3-2}
Given ${{\left\{ {{\mathbf{w}}_{k}} \right\}}_{k\in \mathcal{U}}}, {{\left\{ {{\mathbf{u}}_{t}} \right\}}_{t\in \mathcal{A}}}, {{\left\{ {{\mathbf{v}}_{q}} \right\}}_{q\in \mathcal{T}}}, {{\mathbf{R}}_{s}}$, and $\mathbf{z}$, the reflection coefficients design problem can be recast as the following feasibility search problem 
\begin{subequations}\label{Ref_coeff_optim_V1}
	\begin{align}
        \textrm{find} \quad & {{\left\{ {{\beta }_{t}} \right\}}_{t\in \mathcal{A}}}, \label{Ref_coeff_obj_V1} \\
		\textrm{s.t.:} \quad 
        & (\ref{EH_constraint}), (\ref{opt1_c1})-(\ref{opt1_c3}). \label{Ref_coeff_c1_V1}  
	\end{align}
\end{subequations}
The QoS communication constraint (\ref{opt1_c1}) can be expressed as 
\begin{equation} \label{comm_QoS_Ref}
    \sum\limits_{t=1}^{T}{{{\beta }_{t}}{{c}_{k,t}}}\le \frac{{{a}_{k}}}{\Gamma _{k}^{\mathrm{th}}}-{{b}_{k}},
\end{equation}
where ${{a}_{k}}=\mathrm{Tr}\left( \bm{\overset{\scriptscriptstyle\frown}{\mathbf H}}_{k}{{\mathbf{W}}_{k}} \right)$, ${{b}_{k}}=\sum\limits_{i=1,i\ne k}^{K}{\mathrm{Tr}\left( \bm{\overset{\scriptscriptstyle\frown}{\mathbf H}}_{k}{{\mathbf{W}}_{i}} \right)}+\mathrm{Tr}\left( \bm{\overset{\scriptscriptstyle\frown}{\mathbf H}}_{k}{{\mathbf{R}}_{s}} \right)+\sigma _{k}^{2}$, 
${{c}_{k,t}}=\mathrm{Tr}\left( \bm{\overset{\scriptscriptstyle\frown}{\mathbf H}}_{k,t}{{\mathbf{X}}_{\mathrm{TX}}} \right)$.

Moreover, the BComs and sensing QoS constraints (\ref{opt1_c2})--(\ref{opt1_c3}) can be re-written in unified form as
\begin{equation}\label{gen_QoS_Ref}
\Gamma_{R,l}^{\bar{o},\mathrm{th}}
\Big(\varepsilon_{0,l}+\mathcal{I}_{l}(\boldsymbol{\beta})+\upsilon^{2}\mathcal{C}_{l}(\boldsymbol{\beta})\Big)
\;\le\;
\eta_{l}\,\mathcal{S}_{l}(\boldsymbol{\beta}),
\end{equation}
where 
\[
(\eta_{l},\mathcal{S}_{l})=
\begin{cases}
(\beta_t,\;{\overset{\scriptscriptstyle\frown}{g}}_{t}), & l=t,\ \bar{o}=\mathrm{BComs},\\[1mm]
(\upsilon^{2},\;\sum_{t=1}^{T}\beta_t{\overset{\scriptscriptstyle\frown}{f}}_{t,q}), & l=q,\ \bar{o}=\mathrm{sens}, 
\end{cases}
\]
and the interference/cross-coupling terms are
\[
(\mathcal{I}_{l},\mathcal{C}_{l})=
\begin{cases}
\left(\sum\limits_{i=1,i\neq t}^{T}\beta_i{\overset{\scriptscriptstyle\frown}{g}}_{i,t},\;
\sum\limits_{q=1}^{Q}\sum\limits_{i=1}^{T}\beta_i{\overset{\scriptscriptstyle\frown}{f}}_{i,q,t}\right),
& l=t,\\[3mm]
\left(\sum\limits_{t=1}^{T}\beta_t\tilde{g}_{t,q},\;
\sum\limits_{j=1,j\neq q}^{Q}\sum\limits_{t=1}^{T}\beta_t{\overset{\scriptscriptstyle\frown}{f}}_{t,j,q}\right).
& l=q.
\end{cases}
\]
For $\bar{o} =$BComs, we have 
\begin{align*}
\varepsilon_{0,t}
&= \upsilon^{2}\sum_{q=1}^{Q}\Phi(\mathbf{u}_{t},\mathbf{F}_{q})
+\Phi(\mathbf{u}_{t},\mathbf{H}_{BR})
+\sigma_{R}^{2}\|\mathbf{u}_{t}\|^{2}, \\
\overset{\scriptscriptstyle\frown}{g}_{i,t}
&=\Phi(\mathbf{u}_{t},\mathbf{G}_{i}), \qquad
\overset{\scriptscriptstyle\frown}{g}_{t}=\Phi(\mathbf{u}_{t},\mathbf{G}_{t}),  
\end{align*}
where $\Phi(\mathbf{a},\mathbf{M})
\triangleq \mathbf{a}^{H}\mathbf{M}^{H}\mathbf{R}_{s}\mathbf{M}\mathbf{a}
+\sum_{j=1}^{K}\big|\mathbf{a}^{H}\mathbf{M}^{H}\mathbf{w}_{j}\big|^{2}$.
Moreover, $\overset{\scriptscriptstyle\frown}{f}_{i,q,t}$ is obtained from
$\overset{\scriptscriptstyle\frown}{g}_{i,t}$ by the substitution
$\mathbf{G}_{i}\!\rightarrow\!\mathbf{F}_{i,q}$, with $\mathbf{u}_{t}$, $\mathbf{w}_{j}$, and $\mathbf{R}_{s}$ unchanged. Similarly, for $\bar{o} =$sens, we have 
\begin{align*}
\varepsilon_{0,q}
&= \upsilon^{2}\!\!\sum_{\substack{j=1,~j\neq q}}^{Q}\!\Phi(\mathbf{v}_{q},\mathbf{F}_{j})
+\Phi(\mathbf{v}_{q},\mathbf{H}_{BR})
+\sigma_{R}^{2}\|\mathbf{v}_{q}\|^{2} \\
&-\frac{\upsilon^{2}}{\Gamma_{R,q}^{\mathrm{sens},\mathrm{th}}}\,\Phi(\mathbf{v}_{q},\mathbf{F}_{q}), \qquad 
\tilde{g}_{t,q} =\Phi(\mathbf{v}_{q},\mathbf{G}_{t}), 
\end{align*}
$\overset{\scriptscriptstyle\frown}{f}_{t,j,q}$ follows from $\tilde{g}_{t,q}$ by replacing
$\mathbf{G}_{t}$ with $\mathbf{F}_{t,j}$ (with $\mathbf{v}_{q}$, $\mathbf{w}_{j}$, and $\mathbf{R}_{s}$ unchanged), $\overset{\scriptscriptstyle\frown}{f}_{t,q}$ is the special case obtained by setting $j=q$.

In addition, the EH constraint (\ref{EH_constraint}) can be re-expressed as 
\begin{equation} \label{EH_QoS_Ref}
    {{P}^{\mathrm{EH,th}}}  \le \left( 1-{{\beta }_{t}} \right){{\overset{\scriptscriptstyle\frown}{h}}_{t}}, 
\end{equation}
where $\overset{\scriptscriptstyle\frown}{h}_{t}=\mathrm{Tr}\!\big(\mathbf{H}_{t}\mathbf{X}_{\mathrm{TX}}\big)$.

Accordingly, problem (\ref{Ref_coeff_optim_V1}) can be reformulated as 
\begin{subequations}\label{Ref_coeff_optim_V2}
	\begin{align}
        \textrm{find} \quad & {{\left\{ {{\beta }_{t}} \right\}}_{t\in \mathcal{A}}}, \label{Ref_coeff_obj_V2} \\
		\textrm{s.t.:} \quad 
        & (\ref{comm_QoS_Ref}) - (\ref{EH_QoS_Ref}), \label{Ref_coeff_c1_V2}  \\
        & 0 \le {{\beta }_{t}} \le 1, \forall t\in \mathcal{A}, \label{Ref_coeff_c2_V2} 
	\end{align}
\end{subequations}
which reduces to a linear-constraint feasibility problem, which can be efficiently handled using standard convex optimization solvers, e.g., CVX~\cite{grant2011cvx}.

\subsection{Optimizing the receive beamformers at the reader}\label{sec3-3}
In this section, we design the receive beamforming vectors for BComs, i.e., ${{\left\{ {{\mathbf{u}}_{t}} \right\}}_{t\in \mathcal{A}}}$, and sensing, i.e., ${{\left\{ {{\mathbf{v}}_{q}} \right\}}_{q\in \mathcal{T}}}$, at the reader, while keeping ${{\left\{ {{\mathbf{w}}_{k}} \right\}}_{k\in \mathcal{U}}}$, ${{\left\{ {{\beta }_{t}} \right\}}_{t\in \mathcal{A}}}$, ${{\mathbf{R}}_{s}}$, and $\mathbf{z}$ constant. For each receive beamformer, the problem (\ref{main_problem}) turns into a feasibility problem. Mathematically, such a feasibility problem can be written as: 
\begin{subequations}\label{RX_BF_optim_V1}
	\begin{align}
        \textrm{find} \quad & {{\left\{ {{\mathbf{U}}_{t}} \right\}}_{t\in \mathcal{A}}},{{\left\{ {{\mathbf{V}}_{q}} \right\}}_{q\in \mathcal{T}}}, \label{RX_BF_obj_V1} \\
		\textrm{s.t.:} \quad 
        & \mathrm{Tr}\left( \left( {{{\mathbf{\tilde{G}}}}_{t}}-\Gamma _{R,t}^{\mathrm{BComs,th}}{{{\mathbf{\tilde{F}}}}_{t}} \right){{\mathbf{U}}_{t}} \right)\ge 0,t\in \mathcal{A}, \label{RX_BF_c1_V1}  \\
        & \mathrm{Tr}\left( \left( {{{\mathbf{\tilde{D}}}}_{q}}-\Gamma _{R,t}^{\mathrm{sens,th}}{{{\mathbf{\tilde{C}}}}_{q}} \right){{\mathbf{V}}_{q}} \right)\ge 0,q\in \mathcal{T}, \label{RX_BF_c2_V1} \\
        & \mathrm{Tr}\left( {{\mathbf{U}}_{t}} \right)=1,\quad \mathrm{Tr}\left( {{\mathbf{V}}_{q}} \right)=1, \label{RX_BF_c3_V1}
	\end{align}
\end{subequations}
where the constraints (\ref{RX_BF_c1_V1})--(\ref{RX_BF_c2_V1}) enforce the QoS requirements for BComs and sensing, respectively, whereas (\ref{RX_BF_c3_V1}) imposes the unit-norm normalization on the receive beamformers. Moreover,
${{\mathbf{\tilde{D}}}_{q}}={{\upsilon }^{2}}\left( \mathbf{F}_{q}^{H}{{\mathbf{R}}_{x}}{{\mathbf{F}}_{q}}+\sum\limits_{t=1}^{T}{{{\beta }_{t}}\mathbf{F}_{t,q}^{H}{{\mathbf{R}}_{x}}{{\mathbf{F}}_{t,q}}} \right)$,
${{\mathbf{\tilde{G}}}_{t}}={{\beta }_{t}}\mathbf{G}_{t}^{H}{{\mathbf{R}}_{x}}{{\mathbf{G}}_{t}}$,  
${{\mathbf{\tilde{C}}}_{q}}=\sum\limits_{t=1}^{T}{{{\beta }_{t}}\mathbf{G}_{t}^{H}{{\mathbf{R}}_{x}}{{\mathbf{G}}_{t}}}+{{\upsilon }^{2}}\sum\limits_{r=1,r\ne q}^{Q}{\mathbf{F}_{r}^{H}{{\mathbf{R}}_{x}}{{\mathbf{F}}_{r}}}+{{\upsilon }^{2}}\sum\limits_{r=1,r\ne q}^{Q}{\sum\limits_{t=1}^{T}{{{\beta }_{t}}\mathbf{F}_{t,r}^{H}{{\mathbf{R}}_{x}}{{\mathbf{F}}_{t,r}}}}+\mathbf{H}_{BR}^{H}{{\mathbf{R}}_{x}}{{\mathbf{H}}_{BR}}+\sigma _{R}^{2}{{\mathbf{I}}_{N}}$, and  
${{\mathbf{\tilde{F}}}_{t}}=\sum\limits_{i=1,i\ne t}^{T}{{{\beta }_{i}}\mathbf{G}_{i}^{H}{{\mathbf{R}}_{x}}{{\mathbf{G}}_{i}}}+{{\upsilon }^{2}}\sum\limits_{q=1}^{Q}{\mathbf{F}_{q}^{H}{{\mathbf{R}}_{x}}{{\mathbf{F}}_{q}}}+{{\upsilon }^{2}}\sum\limits_{q=1}^{Q}{\sum\limits_{j=1}^{T}{{{\beta }_{j}}\mathbf{F}_{j,q}^{H}{{\mathbf{R}}_{x}}{{\mathbf{F}}_{j,q}}}}+\mathbf{H}_{BR}^{H}{{\mathbf{R}}_{x}}{{\mathbf{H}}_{BR}}+\sigma _{R}^{2}{{\mathbf{I}}_{N}}$.
To render \eqref{RX_BF_optim_V1} convex, we relax the rank-one constraints
$\mathrm{rank}(\mathbf{U}_t)=1$ and $\mathrm{rank}(\mathbf{V}_q)=1$ and solve the resulting SDR.
If the optimal solutions $\mathbf{U}_t^{\star}$ and/or $\mathbf{V}_q^{\star}$ are not rank-one, feasible rank-one receive beamformers are obtained via Gaussian randomization.
The unit-trace constraints in \eqref{RX_BF_c3_V1} are enforced by the normalizations
$\mathbf{U}_t \leftarrow \mathbf{U}_t/\mathrm{Tr}(\mathbf{U}_t)$ and
$\mathbf{V}_q \leftarrow \mathbf{V}_q/\mathrm{Tr}(\mathbf{V}_q)$.

\subsection{Optimizing the transmit antenna position at the BS}\label{sec3-4}
With the optimized values of the variables ${{\left\{ {{\mathbf{w}}_{k}} \right\}}_{k\in \mathcal{U}}},{{\mathbf{R}}_{s}}$, ${{\left\{ {{\mathbf{u}}_{t}} \right\}}_{t\in \mathcal{A}}}$, ${{\left\{ {{\beta }_{t}} \right\}}_{t\in \mathcal{A}}}$, ${{\left\{ {{\mathbf{v}}_{q}} \right\}}_{q\in \mathcal{T}}}$, the antenna position vector $\mathbf{z}$ can be optimized solving the following problem.
\begin{subequations}\label{FA_optim_V1}
	\begin{align}
        \textrm{find} \quad & \mathbf{z}, \label{FA_obj_V1} \\
		\textrm{s.t.:} \quad 
        & (\ref{opt1_c1}) - (\ref{opt1_c5}), \label{FA_c1_V1}   
	\end{align}
\end{subequations}
which is a non-convex optimization problem. Such a non-convexity stems from the fact that the fractional forms (\ref{opt1_c1}) - (\ref{opt1_c3}), and (\ref{opt1_c5}) are functions of the array response vectors, which explicitly depend on $\mathbf{z}$ through complex exponentials. So the expressions of (\ref{opt1_c1})--(\ref{opt1_c3}) and (\ref{opt1_c5}) become the squared magnitude of a sum of complex exponentials, yielding an oscillatory function, containing cosine terms, which is neither convex nor concave in $\mathbf{z}$.

The communication QoS constraint in (\ref{opt1_c1}) can be expressed as 
\begin{equation} \label{comm_QoS_FA}
    {{\mathcal{S}}_{k}}\ge \Gamma _{k}^{\mathrm{th}}{{\mathcal{I}}_{k}}.
\end{equation}
Note that the terms in $\mathcal{S}_k$ and $\mathcal{I}_k$ are coupled with the matrix $\mathbf{R}_s$, which appears through quadratic forms of the type $\mathbf{x}^{H}\mathbf{R}_s\mathbf{x}$.
To expose a tractable structure, we exploit that $\mathbf{R}_s\succeq \mathbf{0}$ is Hermitian and apply the eigenvalue decomposition
$\mathbf{R}_s=\sum_{r=1}^{M}\lambda_r \mathbf{e}_r\mathbf{e}_r^{H}$, which yields
$\mathbf{x}^{H}\mathbf{R}_s\mathbf{x}
=\sum_{r=1}^{M}\lambda_r \left|\mathbf{e}_r^{H}\mathbf{x}\right|^{2}$.
Accordingly, the sensing-related term can be interpreted as a weighted sum of squared-magnitude effective channels (e.g., $\left|\mathbf{e}_r^{H}\mathbf{h}_k(\mathbf{z})\right|^{2}$), which makes the dependence on the antenna position vector $\mathbf{z}$ explicit and enables expanding the constraints into quadratic expressions in $\mathbf{z}$.
Specifically, the left hand side (LHS) is lower bounded using the Taylor expansions of  trigonometric functions as in the following Lemma.
\begin{lem} \label{lem1}
The intended communication signal ${{\mathcal{S}}_{k}}$ can be lower-bounded as a quadratic expression of $\mathbf{z}$ as follows. 
\begin{equation} \label{FA_comm_useful_LB}
    \mathcal{S}_{k}^{\text{LB}}\left( {{\mathbf{w}}_{k}},\mathbf{z} \right)\triangleq {{\left| \mathbf{w}_{k}^{H}{{\mathbf{h}}_{k}} \right|}^{2}}=\frac{1}{2}{{\mathbf{z}}^{T}}{{\mathbf{R}}_{k}}\mathbf{z}+\mathbf{r}_{k}^{T}\mathbf{z}+{{c}_{k}},
\end{equation}
\end{lem}
\begin{proof}
	See Appendix \ref{appendix1}.
\end{proof} 

Following the analysis in Appendix \ref{appendix1}, the interference $\mathcal{I}_k$ can be denoted as     
\begin{align}
    \mathcal{I}_k
    &=  \sum\limits_{i=1,i\ne k}^{K}{{{\mathcal{S}}_{k}}\left( {{\mathbf{w}}_{i}},\mathbf{z} \right)}+\sum\limits_{r=1}^{M}{{{\lambda }_{r}}{{E}_{k,r}}\left( {{\mathbf{e}}_{r}},\mathbf{z} \right)} \notag \\
       & +\sum\limits_{t=1}^{T}{{{\beta }_{t}}\sum\limits_{r=1}^{M}{{{\lambda }_{r}}{{E}_{k,t,r}}\left( {{\mathbf{e}}_{r}},\mathbf{z} \right)}} \notag \\
       &+\sum\limits_{t=1}^{T}{{{\beta }_{t}}\sum\limits_{j=1}^{K}{{{E}_{j,k,t}}\left( {{\mathbf{w}}_{j}},\mathbf{z} \right)}}+\sigma _{k}^{2}, \label{FA_Interf_comm_V1}
\end{align}
where  
\small
\begin{align}
    & {{E}_{k,r}}={{\left| {{{\tilde{\ell }}}_{k}}\mathbf{e}_{r}^{H}{{{\mathbf{\tilde{h}}}}_{k}} \right|}^{2}}+{{\left| {{{\bar{\ell }}}_{k}} \right|}^{2}}\sum\limits_{m}^{M}{\sum\limits_{{{m}'}}^{M}{{{\nu }_{r,m,{m}'}}\cos \left( {{\varphi }_{r,k,m,{m}'}} \right)}} \notag \\ 
    & +2\sum\limits_{m=1}^{M}{\Re \left\{ {{\left[ {{{\mathbf{\tilde{e}}}}_{k,r}} \right]}_{m}} \right\}\cos \left( {{{\hat{\theta }}}_{k}}{{z}_{m}} \right)+\Im \left\{ {{\left[ {{{\mathbf{\tilde{e}}}}_{k,r}} \right]}_{m}} \right\}\sin \left( {{{\hat{\theta }}}_{k}}{{z}_{m}} \right)} \label{FA_interf_UE_1st_V1},
\end{align}
\normalsize
where $E_{k,t,r}$ can be obtained from $E_{k,r}$ by replacing
$\tilde{\ell}_{k}$ with $\tilde{\omega}_{k,t}$,
$\bar{\ell}_{k}$ with $\bar{\omega}_{k,t}$,
$\tilde{\mathbf{h}}_{k}$ with $\tilde{\mathbf{h}}_{t}$,
$\varphi_{r,k,m,m'}$ with $\varphi_{r,t,m,m'}$,
$\hat{\theta}_{k}$ with $\hat{\theta}_{t}$,
and $\tilde{\mathbf{e}}_{k,r}$ with $\tilde{\mathbf{e}}_{k,t,r}$.
Furthermore, $E_{j,k,t}$ follows from $E_{k,t,r}$ by replacing
$\nu_{r,m,m'}$ with $\varpi_{j,m,m'}$,
$\varphi_{r,t,m,m'}$ with $\varphi_{j,t,m,m'}$,
$\tilde{\mathbf{e}}_{k,t,r}$ with $\tilde{\mathbf{w}}_{j,k,t}$,
and $\mathbf{e}_{r}$ with $\mathbf{w}_{j}$. 
Moreover, we have ${{\nu }_{r,m,{m}'}}=\left| \left[ {{\mathbf{e}}_{r}} \right]_{m}^{*}{{\left[ {{\mathbf{e}}_{r}} \right]}_{{{m}'}}} \right|$, 
${{\varphi }_{r,k,m,{m}'}}={{\hat{\theta }}_{k}}{{z}_{m}}-{{\hat{\theta }}_{k}}{{z}_{{{m}'}}}-\Delta {{\psi }_{r,m,{m}'}}$ , 
$\Delta {{\psi }_{r,m,{m}'}}=\angle {{\left[ {{\mathbf{e}}_{r}} \right]}_{m}}-\angle {{\left[ {{\mathbf{e}}_{r}} \right]}_{{{m}'}}}$, 
${{\mathbf{\tilde{e}}}_{k,r}}={{\bar{\ell }}_{k}}\tilde{\ell }_{k}^{*}\mathbf{\tilde{h}}_{k}^{H}{{\mathbf{e}}_{r}}\mathbf{e}_{r}^{H}$, 
${{\bar{\omega }}_{k,t}}={{g}_{t,k}}\sqrt{\frac{{{\kappa }_{B,t}}}{1+{{\kappa }_{B,t}}}\ell {{\left( \frac{{{d}_{B,t}}}{{{d}_{0}}} \right)}^{-{{\alpha }_{B,t}}}}}$,  ${{\tilde{\omega }}_{k,t}}={{g}_{t,k}}\sqrt{\frac{\ell }{1+{{\kappa }_{B,t}}}{{\left( \frac{{{d}_{B,t}}}{{{d}_{0}}} \right)}^{-{{\alpha }_{B,t}}}}}$,
${{\hat{\theta }}_{t}}=\frac{-2\pi }{\lambda }\cos {{\theta }_{t}}$,
${{\varphi }_{r,t,m,{m}'}}={{\hat{\theta }}_{t}}{{z}_{m}}-{{\hat{\theta }}_{t}}{{z}_{{{m}'}}}-\Delta {{\psi }_{r,m,{m}'}}$, 
${{\mathbf{\tilde{e}}}_{k,t,r}}={{\bar{\omega }}_{k,t}}\tilde{\omega }_{k,t}^{*}\mathbf{\tilde{h}}_{t}^{H}{{\mathbf{e}}_{r}}\mathbf{e}_{r}^{H}$,
${{\varpi }_{j,m,{m}'}}=\left| \left[ {{\mathbf{w}}_{j}} \right]_{m}^{*}{{\left[ {{\mathbf{w}}_{j}} \right]}_{{{m}'}}} \right|$,
${{\varphi }_{j,t,m,{m}'}}={{\hat{\theta }}_{t}}{{z}_{m}}-{{\hat{\theta }}_{t}}{{z}_{{{m}'}}}-\Delta {{\psi }_{j,m,{m}'}}$,
$\Delta {{\psi }_{j,m,{m}'}}=\angle {{\left[ {{\mathbf{w}}_{j}} \right]}_{m}}-\angle {{\left[ {{\mathbf{w}}_{j}} \right]}_{{{m}'}}}$, and 
${{\mathbf{\tilde{w}}}_{j,k,t}}={{\bar{\omega }}_{k,t}}\tilde{\omega }_{k,t}^{*}\mathbf{\tilde{h}}_{t}^{H}{{\mathbf{W}}_{j}}$.

The expressions $E_{k,r}, E_{k,t,t}$ and $E_{j,k,t}$ involve sums of trigonometric functions of $\mathbf{z}$, resulting in non-convexity. Hence, the MM method is re-invoked to upper-bound the corresponding terms via the following trigonometric approximations
\begin{align} \label{UB_sin_cos}
  & \cos \left( x \right)\le \frac{1}{2}{{x}^{2}}+\dot{c}\left( {{x}^{\left( n \right)}} \right)x+\dot{d}\left( {{x}^{\left( n \right)}} \right), \notag \\ 
 & \sin \left( x \right)\le \frac{1}{2}{{x}^{2}}+\ddot{c}\left( {{x}^{\left( n \right)}} \right)x+\ddot{d}\left( {{x}^{\left( n \right)}} \right),  
\end{align}
where $\dot{d}\left( {{x}^{\left( n \right)}} \right)=\cos \left( {{x}^{\left( n \right)}} \right)+{{x}^{\left( n \right)}}\sin \left( {{x}^{\left( n \right)}} \right)-\frac{1}{2}{{\left( {{x}^{\left( n \right)}} \right)}^{2}}$,
$\ddot{d}\left( {{x}^{\left( n \right)}} \right)=\sin \left( {{x}^{\left( n \right)}} \right)-{{x}^{\left( n \right)}}\cos \left( {{x}^{\left( n \right)}} \right)+\frac{1}{2}{{\left( {{x}^{\left( n \right)}} \right)}^{2}}$, 
$\dot{c}\left( {{x}^{\left( n \right)}} \right)=-{{x}^{\left( n \right)}}-\sin \left( {{x}^{\left( n \right)}} \right)$, and 
$\ddot{c}\left( {{x}^{\left( n \right)}} \right)=\cos \left( {{x}^{\left( n \right)}} \right)-{{x}^{\left( n \right)}}$. Afterwards, we substitute (\ref{UB_sin_cos}) into $E_{j,k,t}$ to evaluate the upper-bound of the interference $\mathcal{I}_k$ as 
\begin{equation} \label{FA_UB_comm_interf}
    \mathcal{I}_{k}^{\mathrm{UB}}\left( \mathbf{z} \right)=\frac{1}{2}{{\mathbf{z}}^{T}}{{\mathbf{\hat{R}}}_{k}}\mathbf{z}+\mathbf{\hat{r}}_{k}^{T}\mathbf{z}+{{\hat{c}}_{k}},
\end{equation}
where $\hat{\mathbf R}_{k}
= \bar{\mathbf R}_{k} + \tilde{\mathbf R}_{k} + \bm{\overset{\scriptscriptstyle\smile}{\mathbf R}}_{k} + \bm{\overset{\scriptscriptstyle\frown}{\mathbf R}}_{k}$,   
$\hat{\mathbf r}_{k}^{T} = \bar{\mathbf r}_{k}^{T} + \tilde{\mathbf r}_{k}^{T} + \bm{\overset{\scriptscriptstyle\smile}{\mathbf r}}_{k}^{T}
+ \bm{\overset{\scriptscriptstyle\frown}{\mathbf r}}_{k}^{T}$, and 
${{\hat{c}}_{k}}={{\bar{c}}_{k}}+{{\tilde{c}}_{k}}+{{\overset{\scriptscriptstyle\smile}{c}}_{k}}+{{\overset{\scriptscriptstyle\frown}{c}}_{k}}+\sigma _{k}^{2}$. The relative definitions of the Hessian matrices, linear coefficients, and constant terms are found in Appendix \ref{appendix2}. According to the results of (\ref{FA_comm_useful_LB}) and (\ref{FA_UB_comm_interf}), the communication QoS can be achieved via  
\begin{equation} \label{comm_QoS_FA_final}
    \mathcal{S}_{k}^{\mathrm{LB}}\ge \Gamma _{k}^{\mathrm{th}}\mathcal{I}_{k}^{\mathrm{UB}}.
\end{equation}

Following the results and the analysis of Appendices \ref{appendix1} and \ref{appendix2}, the BComs, sensing, and EH constraints, (\ref{opt1_c2}), (\ref{opt1_c3}), and (\ref{EH_constraint}), can be approximated in similar quadratic expressions. Therefore, the optimization problem of the antenna positions can be recast as   
\begin{subequations}\label{FA_optim_V2}
	\begin{align}
        \textrm{find} \quad & \mathbf{z}, \label{FA_obj_V2} \\
		\textrm{s.t.:} \quad 
        & \mathcal{S}_{k}^{\text{LB}}\ge \Gamma _{k}^{\text{th}}\mathcal{I}_{k}^{\text{UB}},k\in \mathcal{U}, \label{FA_c1_V2} \\ 
        & {{\beta }_{t}}\mathcal{S}_{t}^{\text{LB}}\left( \mathbf{z} \right)\ge \Gamma _{R,t}^{\text{BComs},\text{th}}\mathcal{I}_{t}^{\text{UB}}\left( \mathbf{z} \right),t\in \mathcal{A}, \label{FA_c2_V2} \\ 
        & {{\upsilon }^{2}}\mathcal{S}_{q}^{\text{LB}}\left( \mathbf{z} \right)\ge \Gamma _{R,q}^{\text{sens},\text{th}}\mathcal{I}_{q}^{\text{UB}}\left( \mathbf{z} \right),q\in \mathcal{T},\label{FA_c3_V2} \\
        & \left( 1-{{\beta }_{t}} \right)\mathcal{P}_{t}^{\text{in,LB}}\left( \mathbf{z} \right)\ge {{\Phi }^{-1}}\left( \rho  \right), \nonumber \\
        & 0<{{\beta }_{t}}<1, t\in \mathcal{A}, \label{FA_c4_V2} \\
        & (\ref{opt1_c4}),
	\end{align}
\end{subequations}
which is a convex optimization problem that can be solved using standard convex optimization tools such as the CVX toolbox \cite{grant2011cvx}.


\section{Simulation Results and Discussion}\label{sec4}
\begin{table}[t]
\centering
\scriptsize
\caption{Simulation Parameters}
\label{tab:sim_params}
\setlength{\tabcolsep}{3pt}
\resizebox{\columnwidth}{!}{%
\begin{tabular}{l l l l}
\hline
\textbf{Symbol} & \textbf{Value} & \textbf{Symbol} & \textbf{Value} \\
\hline
$B$                     & $10$ MHz            & $N_f$                                            & $10$ dB \\
$M$                     & 16                  & $N_1 = N_2$                                      & 4 \\
$K$                     & 3                   & $T$                                              & 2 \\
$Q$                     & 2                   & $\upsilon^{2}$                                   & 1 \\
Rician factor (BS-$U_k$)   & 10 dB               & Rician factor (BS-$A_t$, BS-$T_q$, BS-reader)      & 6 dB \\
Rician factor (reader-$A_t$, reader-$T_q$)       & 5 dB       & Rician factor ($U_k$-$A_t$, $A_t$-$T_q$)   & 3 dB  \\
Path-loss exponent (all links)         & 2.2               & Path-loss exponent (BS-reader)        & 2.4  \\
$\Gamma _{k}^{\mathrm{th}} = \Gamma _{R,t}^{\mathrm{BComs},\mathrm{th}} = \Gamma _{R,q}^{\mathrm{sens},\mathrm{th}}$      & 1       & Activation input power of EH ($a,b,c$) & 2.463, 1.635, 0.826   \\
EH linear input, $\rho$                                   & -25 dBm & Path-loss at the reference
distance                 & -30 dB      \\
\hline
\end{tabular}%
} 
\end{table}
The simulation results are presented to evaluate the performance of the proposed FA-aided bistatic ISABC network. The 3GPP urban micro model is adopted to characterize the path-loss values at a carrier frequency of $f_c = 3.5~\mathrm{GHz}$ \cite[Table~B.1.2.1]{access2010further}. The AWGN power $\sigma^2$ is computed as $\sigma^2 = 10\log_{10}(N_0 B N_f)~\mathrm{dBm}$, where $N_0 = -174~\mathrm{dBm/Hz}$ denotes the noise power spectral density, $B$ is the system bandwidth, and $N_f$ is the receiver noise figure. Unless stated otherwise, the main simulation parameters are summarized in Table~\ref{tab:sim_params}, and all results are obtained over $10^3$ Monte Carlo realizations. To emulate smart home/city deployments, the BS and the reader are located at $\{0,0\}$ and $\{12,0\}$, respectively. The communication users, tags, and radar targets  are randomly distributed within a disks centered at $\{55,0\}$, $\{8,-4\}$, $\{8,4\}$ of  radius $5~\mathrm{m}$, $3~\mathrm{m}$, $3~\mathrm{m}$, respectively \cite{galappaththige2022ris}.

For comparative evaluation, we consider the following benchmark schemes: 1) \textit{FPA}, where the BS antennas are deployed at fixed locations; 2) \textit{ZF}, where the BS downlink beamforming for the communication users is constructed based on the BS–user channels to suppress inter-user interference; 3) \textit{Proposed, Com-BCom}, which follows the proposed framework but drops the sensing QoS constraint ($\mathbf{R}_s = \mathbf{0}$); and 4) \textit{Proposed, sensing only}, which retains only the sensing QoS constraint and relaxes all communication/backscatter-related constraints, providing a lower-bound reference on the minimum transmit power when sensing is the main design objective.

\begin{figure}[t]
\centering{\includegraphics[width=0.90\columnwidth]{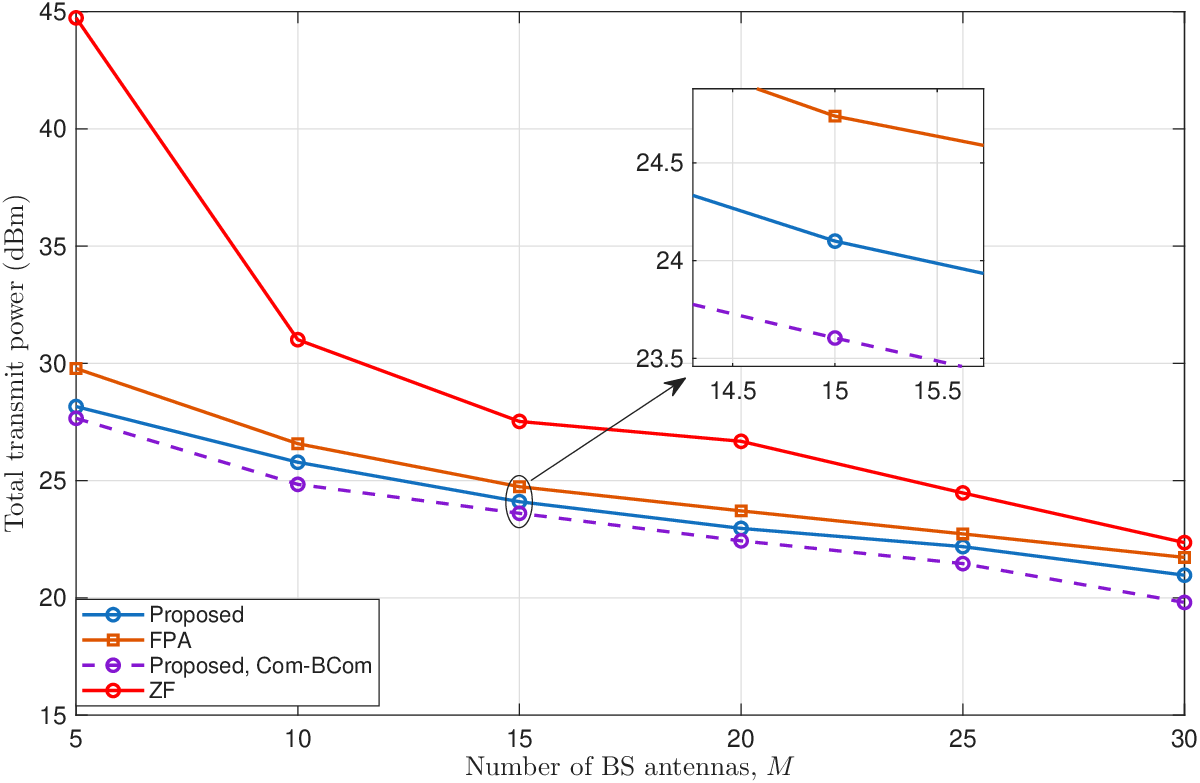}}
\caption{Transmit power versus the number of BS antennas, $M$.}\label{BS_antennas}
\end{figure}
Fig.~\ref{BS_antennas} shows the impact of increasing the number of BS antennas on the total transmit power. It is clear that increasing the number of BS antennas monotonically reduces the total transmit power because a larger $M$ provides higher array gain and more spatial DoFs. Thus, the BS is able to better suppress interference and shape the composite signal observed at the users and the reader. Moreover, the proposed scheme consistently outperforms the FPA baseline, since the FA position primarily induces controllable phase variations that can be exploited to further improve spatial selectivity beyond conventional beamforming. For instance, at $M=15$, the proposed design reduces the required transmit power by 0.64 dB and 3.42 dB compared with FPA and ZF, respectively, corresponding to 13.7\% and 54.5\% savings in linear scale. 
In addition, the performance gap between “Proposed, Com-BCom” and the full “Proposed” design directly reflects the extra transmit power that must be spent to additionally satisfy the sensing QoS in the joint sensing–communication–backscatter formulation.

\begin{figure}[t]
\centering{\includegraphics[width=0.90\columnwidth]{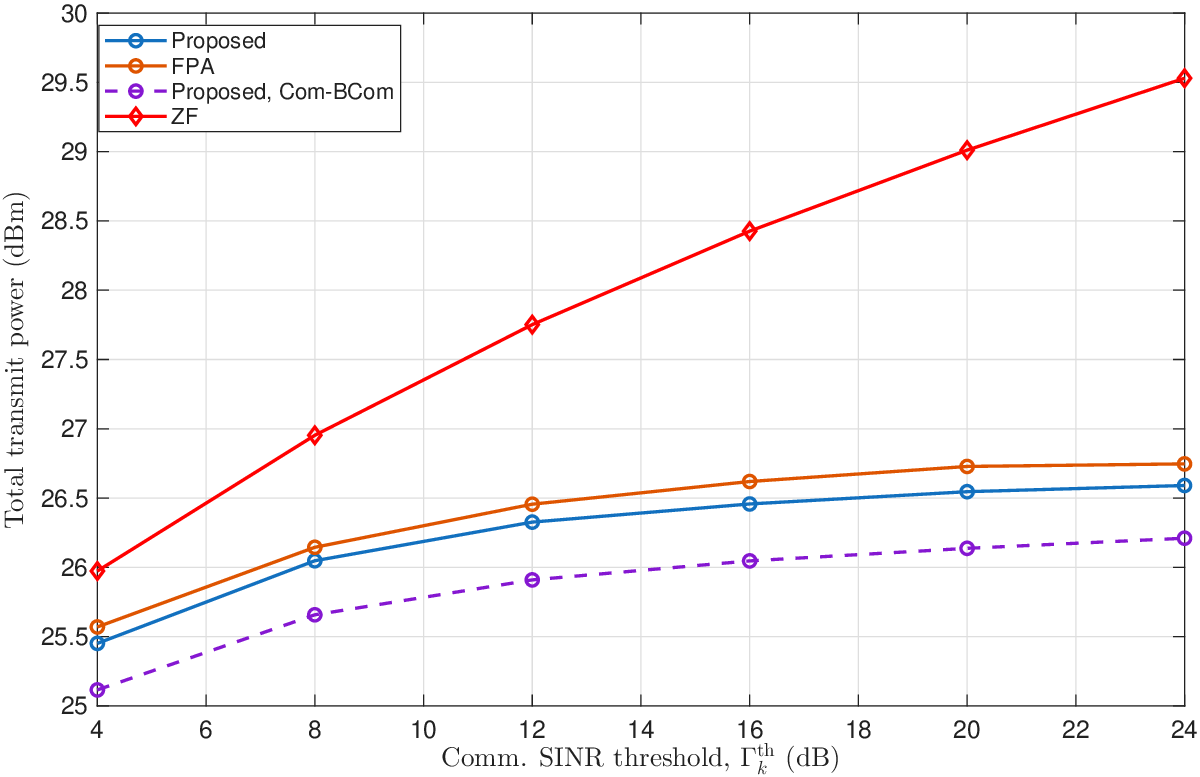}}
\caption{Total transmit power versus $\Gamma_k^{\mathrm{th}}$.}\label{TXpower_commTH}
\end{figure}

Next, Fig.~\ref{TXpower_commTH} investigates the impact of the communication SINR threshold, $\Gamma_k^{\mathrm{th}}$, on the total transmit power. As $\Gamma_k^{\mathrm{th}}$  increases, all the schemes require more transmit power, since the feasible set shrinks and stronger beams must be formed toward the users. The “Proposed, Com-BCom” yields the lowest transmit power because dropping the sensing constraint enlarges the optimization search space. Meanwhile, the proposed scheme consistently outperforms FPA, confirming that exploiting antenna repositioning provides additional spatial DoFs that translate into tangible power savings. The ZF baseline exhibits the highest and steepest power growth with $\Gamma_k^{\mathrm{th}}$, since it enforces interference nulling solely in the BS–user channel subspace, without accounting for the additional coupling introduced by the sensing/backscatter links. Consequently, once the SINR targets become more stringent, the remaining DoFs are insufficient to efficiently satisfy all the constraints simultaneously, and the BS  compensates by increasing the transmit power.

\begin{figure}[t]
\centering{\includegraphics[width=0.90\columnwidth]{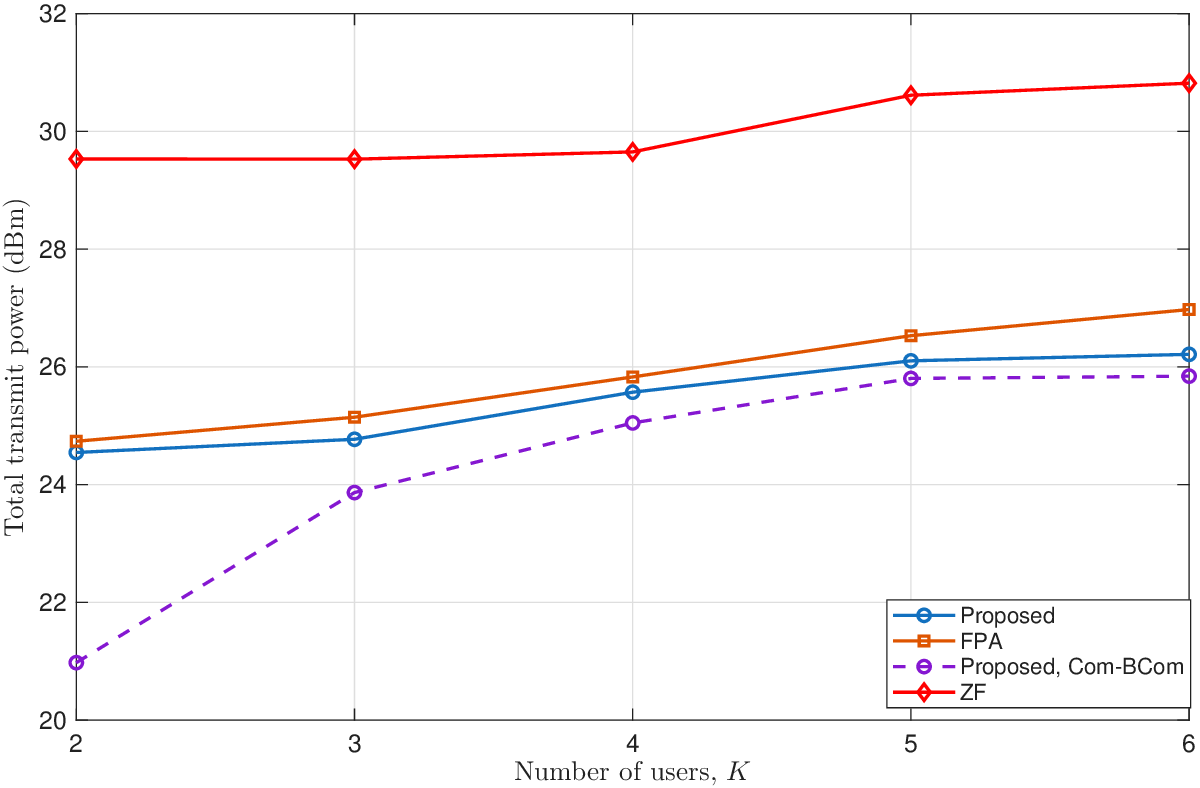}}
\caption{Impact of increasing the number of users on the total transmit power.}\label{Num_users_TXpower}
\end{figure}
The impact of increasing the number of communication users on the total transmit power is illustrated in Fig.~\ref{Num_users_TXpower}, where all schemes show monotonic increment with $K$. This behavior stems from the fact that adding more users introduces additional multi-user interference, thereby shrinking the feasible set of the joint power-minimization problem. The proposed scheme consistently requires less power than the FPA baseline, indicating that FA repositioning provides extra DoFs to better manage the growing interference as $K$ increases. Moreover, the “Proposed, Com-BCom” benchmark yields the lowest power because it relaxes the sensing QoS constraint, allowing more resources to be allocated to satisfying the communication and backscatter requirements. In contrast, the ZF baseline remains significantly more power-hungry as 
$K$ grows, leading to inefficient power usage in multi-user regimes.

\begin{figure}[t]
\centering{\includegraphics[width=0.90\columnwidth]{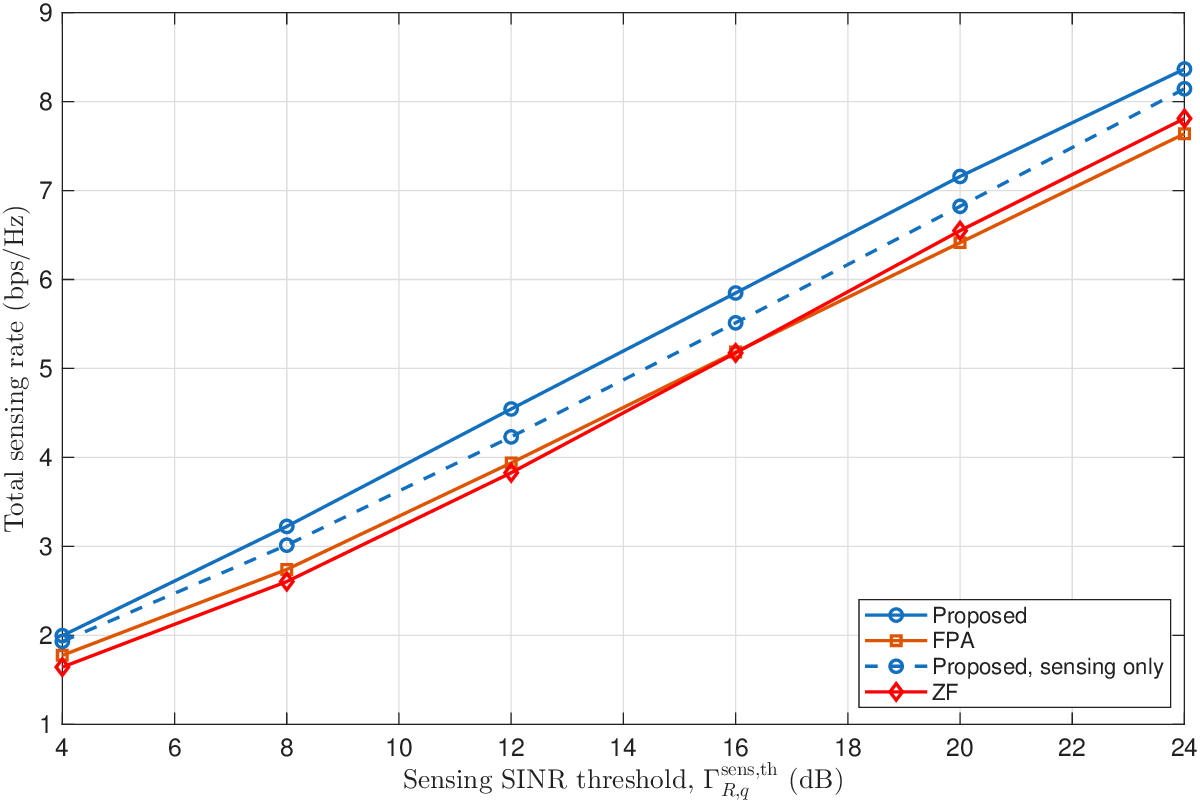}}
\caption{Total sensing rate versus $\Gamma_{R,q}^{\mathrm{sens, th}}$.}\label{sensRate_sensTH}
\end{figure}
Next, Fig.~\ref{sensRate_sensTH} depicts the impact of the sensing-SINR threshold on the sensing rate. Clearly, the sensing rate grows with $\Gamma_{R,q}^{\mathrm{sens, th}}$ for all schemes, but with different relative behaviors for the FPA and the ZF. It is noteworthy that the  ZF design enforces a rigid user-interference-nulling, which is not sensing-aware and can  reduce the energy directed toward the sensing directions. When $\Gamma_{R,q}^{\mathrm{sens, th}} > 16$ dB, the ZF-based design compensates for its less favorable beam pattern by significantly increasing the transmit power, which boosts the received echo power at the reader and causes its sensing rate to eventually surpass that of the FPA, even though it remains less power-efficient overall.

\begin{figure}[t]
\centering{\includegraphics[width=0.90\columnwidth]{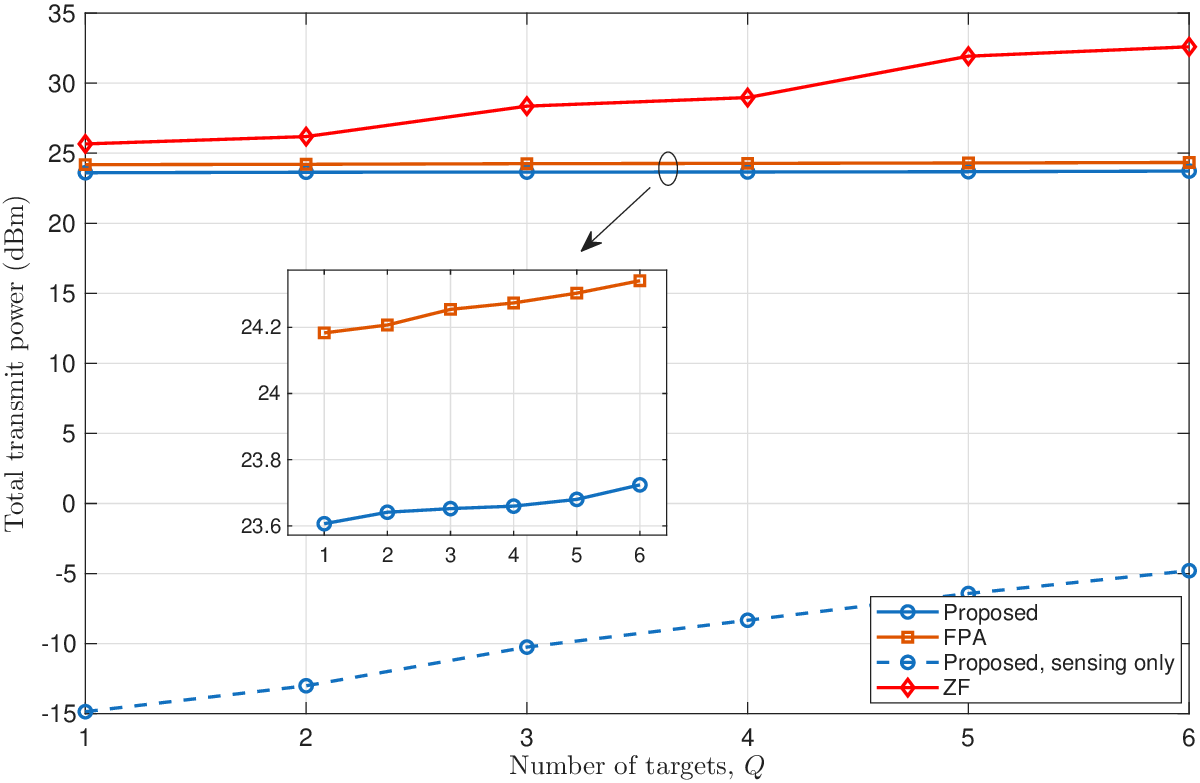}}
\caption{Impact of increasing the number of radar target on the total transmit power.}\label{Num_targets_TXpower}
\end{figure}
Moreover, Fig.~\ref{Num_targets_TXpower} illustrates that increasing the number of sensing targets  
leads to a gradual rise in the total transmit power. This mild increment emphasizes that the constraints (\ref{opt1_c1}-\ref{opt1_c2}) determine the operating power level, and the additional sensing constraints mainly consume the available sensing margin rather than fundamentally changing the optimal power budget. The proposed scheme consistently requires less power than the FPA  baseline, confirming that FA repositioning provides extra spatial degrees of freedom to accommodate the growing set of sensing constraints more efficiently. By contrast, the ZF benchmark becomes increasingly power-hungry as $Q$ grows, since the BS must raise its transmit power to keep the ZF constraints feasible as more targets are imposed.

\begin{figure}[t]
\centering{\includegraphics[width=0.90\columnwidth]{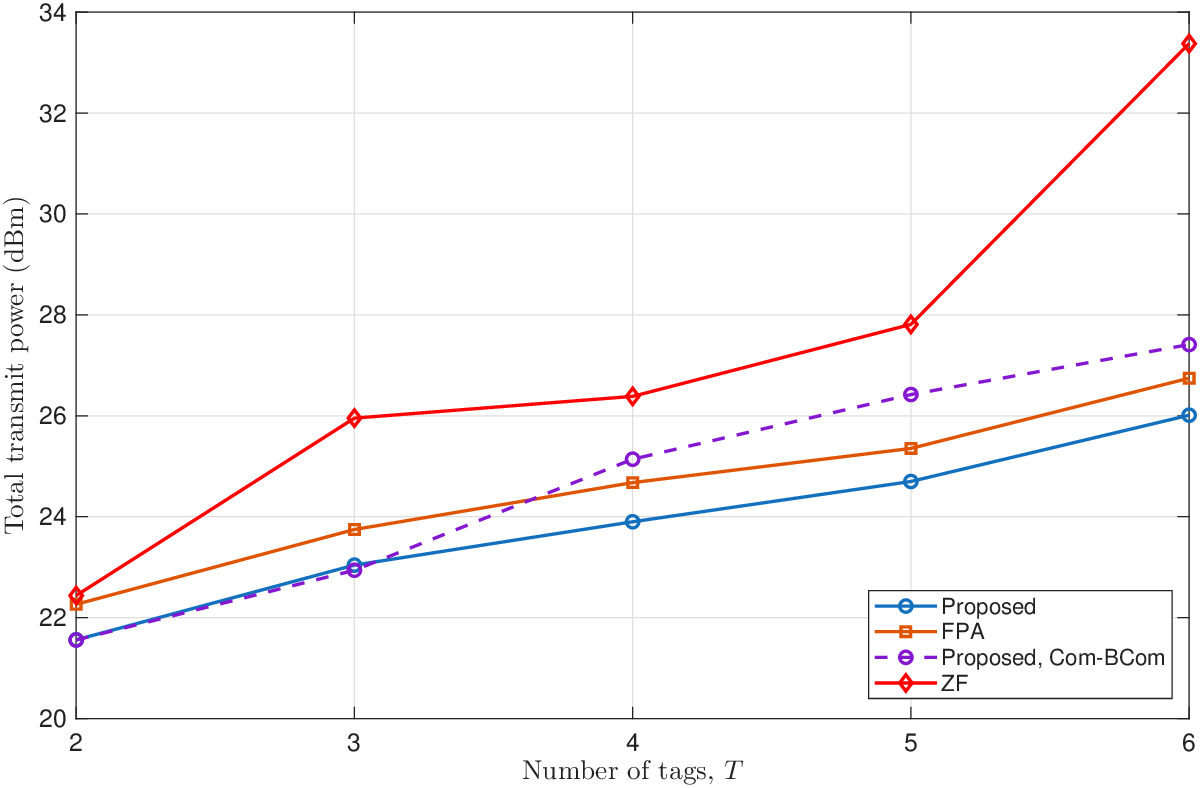}}
\caption{Total transmit power versus the number of tags.}\label{Num_tags_TXpower}
\end{figure}
Fig.~\ref{Num_tags_TXpower} demonstrates that the total transmit power increases with the number of tags. Increasing $T$ introduces additional coupled constraints of the backscatter links, users-tags interference, while the reader must decode BCom and sense targets. This requirement tightens the joint QoS feasibility and exacerbates multi-link interference. Therefore, the BS is forced to radiate more power to simultaneously satisfy these demands. The proposed FA-enabled design consistently achieves the lowest power thanks to extra DoFs that mitigates the growing multi-tag interference. Considering the Proposed and FPA algorithms, we note that, the BS transmits not only information beams $\{ \mathbf{w}_k \}$ but also a dedicated sensing signal with covariance $\mathbf{R}_s$. That sensing component can also contribute to the incident RF power at the tags, hence supporting EH, while being spatially shaped to limit harm to the users/reader, which becomes increasingly valuable when $T$ is large. Conversely, in the “Proposed, Com-BCom” benchmark, when the number of tags is further increased, e.g. $T \ge 4$,  the tags are powered mainly through the information beams. That typically increases the multiuser/backscatter interference and forces a larger transmit power to keep a feasible service for the users/tags.

\begin{figure}[t]
\centering{\includegraphics[width=0.90\columnwidth]{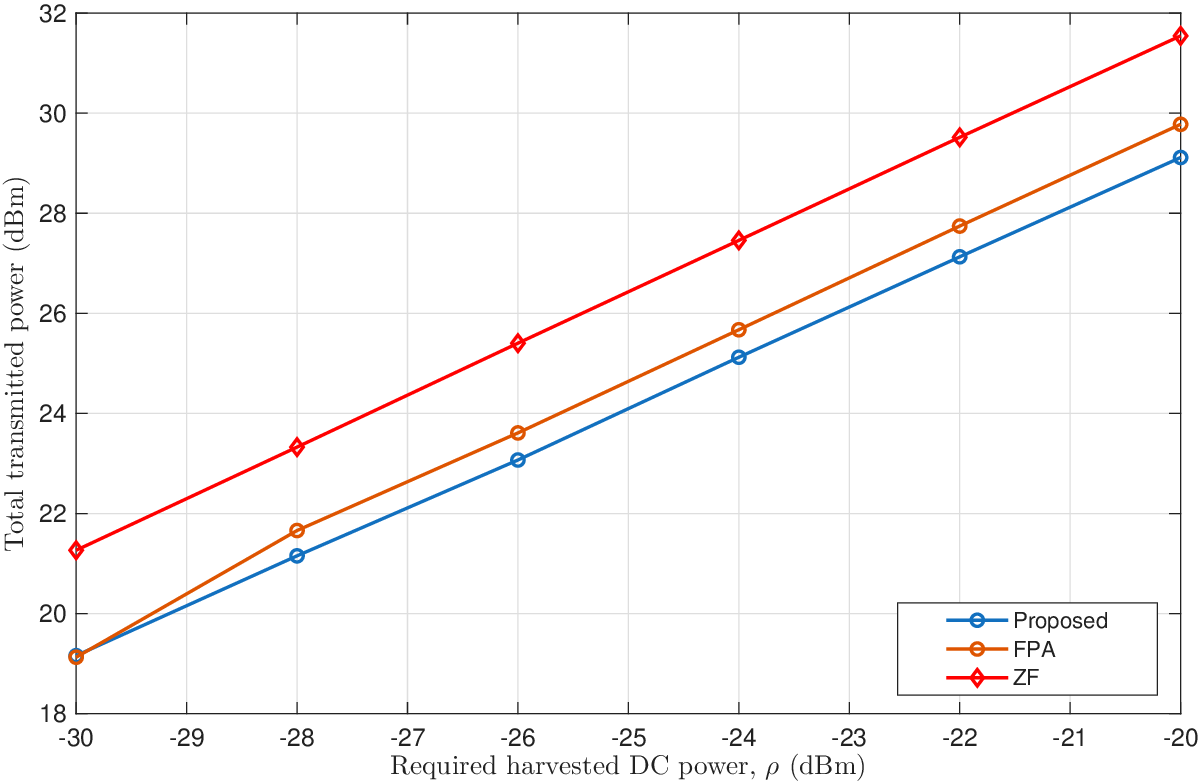}}
\caption{The total transmit power versus $\rho$ .}\label{TXpower_EH_TH}
\end{figure}
Fig.~\ref{TXpower_EH_TH} depicts the impact of $\rho$ on the total transmit power. Increasing the required harvested DC power $\rho$ monotonically increases the minimum total transmit power for all schemes. Since the passive tags rely on the incident RF signal to meet their harvesting requirement, the BS must radiate more power to guarantee sufficient incident energy at the tags. The proposed method consistently requires the lowest transmit power compared with FPA and ZF. For example, at $\rho = -24$ dBm, the proposed scheme reduces the required transmit power by about $11.9\%$ relative to FPA and by $41.7\%$ relative to ZF.

\begin{figure}[t]
\centering{\includegraphics[width=0.90\columnwidth]{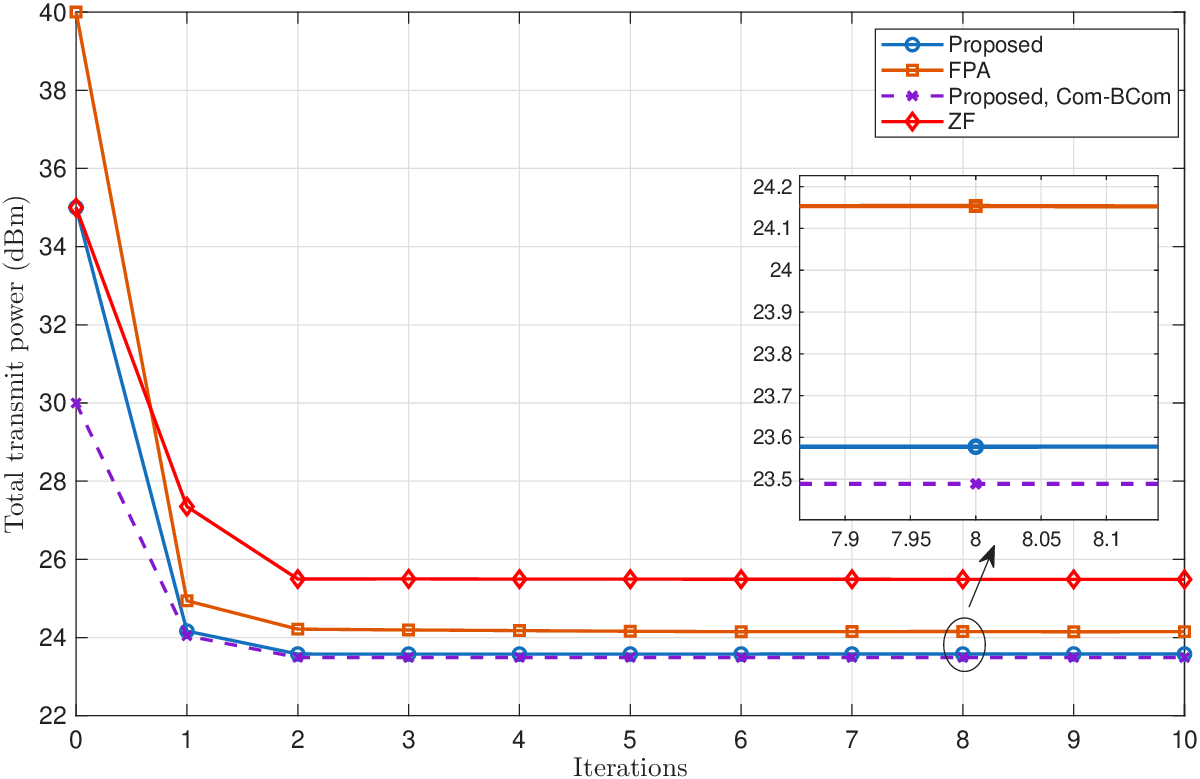}}
\caption{Convergence of the proposed algorithm.}\label{converg}
\end{figure}
Finally, Fig.~\ref{converg} verifies the rapid and stable behavior of the proposed  architecture, where the alternating updates can quickly exploit the additional spatial DoFs offered by FA to reduce the required transmit power. Starting from an initial feasible point, all curves exhibit a steep objective reduction within the first few iterations and then saturate, indicating that the block-coordinate procedure converges to a stationary point of the joint design that minimizes the total transmit power. Importantly, the steady-state gap between the proposed FA-enabled scheme and the FPA baseline is practically meaningful. Specifically, the proposed approach converges to 23.58 dBm, whereas FPA converges to 
24.15 dBm, which represents a gap of 32.3 mW reduction, i.e., about a 12.4\% transmit-power saving relative to the FPA setup. Moreover, compared to ZF, the proposed design achieves 35.6\% transmit-power savings, highlighting that user-centric ZF beamforming cannot efficiently accommodate the coupled sensing and BComs requirements. This directly reflects the benefit of optimizing the FA position in addition to beamforming and covariance design under the multi-service constraints.


\section{Conclusion}\label{sec5}
In this paper, we studied a multi-service FA-enabled bistatic ISABC  architecture and formulated a unified transmit-power minimization problem that jointly optimizes BS precoding, sensing covariance, tag reflection, reader beamforming, and FA positions. We then developed an efficient AO-based block-coordinate algorithm that satisfies heterogeneous QoS requirements for communication, BCom, sensing, and EH constraints. Numerical results show fast convergence and consistent power savings over benchmarks, with gains that persist as the number of antennas, users, tags, and sensing targets increases. 

The results also reveal inherent multi-service trade-offs. In particular, imposing a more stringent communication SINR requirement increases the minimum transmit power while enhancing the achievable rate. In contrast, a higher EH activation threshold demands additional power to energize the tags and may degrade the communication rate because  spatial and power resources are reallocated to satisfy the EH constraint. 
Moreover, comparing the full design with the “Com-BCom” variant quantifies the extra power required to guarantee the sensing QoS. Meanwhile, the relatively weak dependence of the optimum power on the sensing-SINR threshold in several regimes suggests that communication, BCom, and EH constraints often dominate the power budget and that sensing can be supported mainly through waveform/beam-pattern shaping rather than substantial power escalation.
Overall, these findings indicate that FA repositioning provides a practical mechanism to balance energy efficiency, communication throughput, and sensing quality in coupled sensing–communication–backscatter systems.

\appendices
\section{Lower-bounding the intended communication signal (\ref{FA_comm_useful_LB})}\label{appendix1}
The intended communication signal of (\ref{SINR_UE}), $\mathcal{S}_k$, can be re-expressed in terms of FA position $\mathbf{z}$ as 
\begin{equation} \label{FA_comm_useful_V1}
    {{\mathcal{S}}_{k}}\left( {{\mathbf{w}}_{k}},\mathbf{z} \right) = {{\left| {{{\bar{\ell }}}_{k}}{{{\dot{w}}}_{k}} \right|}^{2}}+{{\left| Y \right|}^{2}}+2\Re \left\{ {{X}^{*}}Y \right\},
\end{equation}
 where $\mathcal{S}_k \triangleq {{\left| X+Y \right|}^{2}}={{\left| X \right|}^{2}}+{{\left| Y \right|}^{2}}+2\Re \left\{ {{X}^{*}}Y \right\}$, with $X\triangleq {{\bar{\ell }}_{k}}\mathbf{w}_{k}^{H}{{\mathbf{a}}_{B}}\left( {{\theta }_{k}},\mathbf{z} \right)$ and $Y\triangleq {{\tilde{\ell }}_{k}}\mathbf{w}_{k}^{H}{{\mathbf{\tilde{h}}}_{k}}$. In addition, we have  
 ${{\bar{\ell }}_{k}}=\sqrt{\frac{{{\kappa }_{B,k}}\ell }{1+{{\kappa }_{B,k}}}{{\left( \frac{{{d}_{B,k}}}{{{d}_{0}}} \right)}^{-{{\alpha }_{B,k}}}}}$,   ${{\tilde{\ell }}_{k}}=\sqrt{\frac{\ell }{1+{{\kappa }_{B,k}}}{{\left( \frac{{{d}_{B,k}}}{{{d}_{0}}} \right)}^{-{{\alpha }_{B,k}}}}}$. The expression $\mathbf{w}_{k}^{H}{{\mathbf{a}}_{B}}\left( {{\theta }_{k}},\mathbf{z} \right)$ can be denoted as 
 \begin{equation} \label{FA_expansion_comm1}
      {{\dot{w}}_{k}}\left( {{\theta }_{k}},\mathbf{z} \right) \triangleq \mathbf{w}_{k}^{H}{{\mathbf{a}}_{B}}=\sum\limits_{m=1}^{M}{\left[ {{\mathbf{w}}_{k}} \right]_{m}^{*}{{e}^{j{{{\hat{\theta }}}_{k}}{{z}_{m}}}}},
 \end{equation}
where ${{\hat{\theta }}_{k}}=\frac{-2\pi }{\lambda }\cos {{\theta }_{k}}$. To evaluate the first term of (\ref{FA_comm_useful_V1}), we define 
\begin{equation} \label{FA_expansion_comm2}
  {{\left| {{{\dot{w}}}_{k}} \right|}^{2}} 
  =\sum\limits_{m=1}^{M}{\sum\limits_{{m}'=1}^{M}{{{\varpi }_{k,m,{m}'}}.\cos \left( {{\varphi }_{k,m,{m}'}} \right)}} 
  \triangleq g_{1,k}^{\cos }\left( {{\mathbf{w}}_{k}},\mathbf{z} \right),  
\end{equation}
where ${{\varpi }_{k,m,{m}'}}=\left| \left[ {{\mathbf{w}}_{k}} \right]_{m}^{*}{{\left[ {{\mathbf{w}}_{k}} \right]}_{{{m}'}}} \right|$,  
${{\varphi }_{k,m,{m}'}}={{\hat{\theta }}_{k}}{{z}_{m}}-{{\hat{\theta }}_{k}}{{z}_{{{m}'}}}-\Delta {{\psi }_{k,m,{m}'}}$,  and $\Delta {{\psi }_{k,m,{m}'}}=\angle {{\left[ {{\mathbf{w}}_{k}} \right]}_{m}}-\angle {{\left[ {{\mathbf{w}}_{k}} \right]}_{{{m}'}}}$. Moreover, the cross-term $2\Re \left\{ {{X}^{*}}Y \right\}$, of (\ref{FA_comm_useful_V1}), can be expressed as
\begin{multline} \label{FA_expansion_comm3}
    2\Re \left\{ {{X}^{*}}Y \right\} = 2\Re \left\{ \sum\limits_{m=1}^{M}{{{\left[ {{{\mathbf{\tilde{w}}}}_{k}} \right]}_{m}}}{{e}^{j{{{\hat{\theta }}}_{k}}{{z}_{m}}}} \right\} \\ 
    =2\sum\limits_{m=1}^{M}{\Re \left\{ {{\left[ {{{\mathbf{\tilde{w}}}}_{k}} \right]}_{m}} \right\}\cos \left( {{{\hat{\theta }}}_{k}}{{z}_{m}} \right)+\Im \left\{ {{\left[ {{{\mathbf{\tilde{w}}}}_{k}} \right]}_{m}} \right\}\sin \left( {{{\hat{\theta }}}_{k}}{{z}_{m}} \right)} \\
    \triangleq {{g}_{2,k}}\left( {{\mathbf{w}}_{k}},\mathbf{z} \right),
\end{multline}
where ${{\mathbf{\tilde{w}}}_{k}}={{\bar{\ell }}_{k}}\tilde{\ell }_{k}^{*}\mathbf{\tilde{h}}_{k}^{H}{{\mathbf{W}}_{k}}$. By plugging (\ref{FA_expansion_comm2}) and (\ref{FA_expansion_comm3}), the expression (\ref{FA_comm_useful_V1}) can be recast as 
\begin{equation} \label{FA_comm_useful_V2}
    {{\mathcal{S}}_{k}} = {{\left| {{{\bar{\ell }}}_{k}} \right|}^{2}}g_{1,k}^{\cos } + {{g}_{2,k}} + {{\left| {{{\tilde{\ell }}}_{k}}\mathbf{w}_{k}^{H}{{{\mathbf{\tilde{h}}}}_{k}} \right|}^{2}}.
\end{equation}
However, the expressions (\ref{FA_comm_useful_V2}) is very complex because of the sum of trigonometric functions in the variable $\mathbf{z}$, which leads to non-convexity. To handle this issue, we employ MM method, where trigonometric functions are approximated via second-order Taylor expansion, resulting in a quadratic surrogate functions as in the following. 
\begin{align} \label{LB_sin_cos}
    \cos \left( x \right) & \ge -\frac{1}{2}{{x}^{2}}+\dot{a}\left( {{x}^{\left( n \right)}} \right)x+\dot{b}\left( {{x}^{\left( n \right)}} \right),  \notag \\ 
    \sin \left( x \right) & \ge -\frac{1}{2}{{x}^{2}}+\ddot{a}\left( {{x}^{\left( n \right)}} \right)x+\ddot{b}\left( {{x}^{\left( n \right)}} \right),  
\end{align}
where the superscript $\left( n \right)$ denotes the index of $n$-th iteration. Moreover, we have $\dot{b}\left( {{x}^{\left( n \right)}} \right)=\cos \left( {{x}^{\left( n \right)}} \right)+{{x}^{\left( n \right)}}\sin \left( {{x}^{\left( n \right)}} \right)-\frac{1}{2}{{\left( {{x}^{\left( n \right)}} \right)}^{2}}$, 
$\ddot{b}\left( {{x}^{\left( n \right)}} \right)=\sin \left( {{x}^{\left( n \right)}} \right)-{{x}^{\left( n \right)}}\cos \left( {{x}^{\left( n \right)}} \right)-\frac{1}{2}{{\left( {{x}^{\left( n \right)}} \right)}^{2}}$, 
$\dot{a}\left( {{x}^{\left( n \right)}} \right)={{x}^{\left( n \right)}}-\sin \left( {{x}^{\left( n \right)}} \right)$, and 
$\ddot{a}\left( {{x}^{\left( n \right)}} \right)={{x}^{\left( n \right)}}+\cos \left( {{x}^{\left( n \right)}} \right)$. 
Accordingly, the first term of (\ref{FA_comm_useful_V2}) is bounded as  
\begin{align}
g_{1,k}^{\cos}
&\approx \sum_{m=1}^{M}\sum_{m'=1}^{M} \varpi_{k,m,m'}\Big[
\underbrace{-\tfrac{1}{2}\hat{\theta}_{k}^{2} z_{m}^{2}
-\tfrac{1}{2}\hat{\theta}_{k}^{2} z_{m'}^{2}
+\hat{\theta}_{k}^{2} z_{m} z_{m'}}_{\text{quadratic term}}
\notag\\[-1pt]&\quad
+\,\underbrace{\bar{\varepsilon}_{2,k,m,m'}(z_{m}-z_{m'})}_{\text{linear term}}
+ \underbrace{\varepsilon_{2,k,m,m'}}_{\mathclap{\text{constant}}}\Big], \label{g1kcos_LB}
\end{align}
where ${{\varepsilon }_{2,k,m,{m}'}}=\dot{b}\left( \varphi _{k,m,{m}'}^{\left( n \right)} \right)-\dot{a}\left( \varphi _{k,m,{m}'}^{\left( n \right)} \right)\Delta {{\psi }_{k,m,{m}'}}-\frac{1}{2}\Delta \psi _{k,m,{m}'}^{2}$ and ${{\bar{\varepsilon }}_{2,k,m,{m}'}}=\left( \dot{a}\left( \varphi _{k,m,{m}'}^{\left( n \right)} \right)+\Delta {{\psi }_{k,m,{m}'}} \right){{\hat{\theta }}_{k}}$.  
The first term in (\ref{g1kcos_LB}) can be re-expressed as $\frac{1}{2}{{\mathbf{z}}^{T}}\mathbf{R}_{g1}^{\cos }\mathbf{z}$ by finding the Hessian matrix  $\mathbf{R}_{g1}^{\cos }$ as   
\begin{equation} \label{g1kcos_LB_Hessian}
    \mathbf{R}_{g1}^{\cos }=2\hat{\theta }_{k}^{2}\left| {{\pmb{\omega }}_{k}} \right|{{\left| {{\pmb{\omega }}_{k}} \right|}^{T}}-2\hat{\theta }_{k}^{2}\left( \sum\limits_{m}^{M}{\left| {{\left[ {{\mathbf{w}}_{k}} \right]}_{m}} \right|} \right)\text{diag}\left( \left| {{\pmb{\omega }}_{k}} \right| \right),
\end{equation}
where $\left| {{\pmb{\omega }}_{k}} \right|={{\left[ \left| {{\left[ {{\mathbf{w}}_{k}} \right]}_{1}} \right|,...,\left| {{\left[ {{\mathbf{w}}_{k}} \right]}_{M}} \right| \right]}^{T}}$.
The linear part of (\ref{g1kcos_LB}) can be re-written as ${{\mathbf{r}}^{T}}\mathbf{z}=\sum\limits_{p=1}^{M}{{{r}_{p}}{{z}_{p}}}$, where the $p-$th coefficient of vector $\mathbf{r}$ is given by  
\begin{multline}\label{g1kcos_LB_linear}
    \bigl[\mathbf{r}_{g1}^{\cos}\bigr]_p
    = \hat{\theta}_{k}\sum_{m'=1}^{M} \varpi_{k,p,m'}\!\left(\dot{a}\!\left(\varphi_{k,p,m'}^{(n)}\right)+\Delta \psi_{k,p,m'}\right) \\
    \hspace{3.2em}
    - \hat{\theta}_{k}\sum_{m'=1}^{M} \varpi_{k,m',p}\!\left(\dot{a}\!\left(\varphi_{k,m',p}^{(n)}\right)+\Delta \psi_{k,m',p}\right), 
\end{multline}
where ${{\varpi }_{k,p,{m}'}}=\left| \left[ {{\mathbf{w}}_{k}} \right]_{p}^{*}{{\left[ {{\mathbf{w}}_{k}} \right]}_{{{m}'}}} \right|$ and ${{\varpi }_{k,{m}',p}}=\left| \left[ {{\mathbf{w}}_{k}} \right]_{{{m}'}}^{*}{{\left[ {{\mathbf{w}}_{k}} \right]}_{p}} \right|$.

Similarly, the second term of (\ref{FA_comm_useful_V2}) can be  bounded as 
\small
\begin{align}
g_{2,k}
&= -\sum_{m=1}^{M}
   \underbrace{\Big(\Re\{[\tilde{\mathbf{w}}_{k}]_{m}\}
   + \Im\{[\tilde{\mathbf{w}}_{k}]_{m}\}\Big)(\hat{\theta}_{k} z_{m})^{2}}_{\mathclap{\text{quadratic term}}}
\notag\\
& + 2\sum_{m=1}^{M}
   \underbrace{\Big(\Re\{[\tilde{\mathbf{w}}_{k}]_{m}\}\,\dot{a}(\hat{\theta}_{k} z_{m}^{(n)})
   + \Im\{[\tilde{\mathbf{w}}_{k}]_{m}\}\,\ddot{a}(\hat{\theta}_{k} z_{m}^{(n)})\Big)\hat{\theta}_{k} z_{m}}_{\mathclap{\text{linear term}}}
\notag\\
& + 2\sum_{m=1}^{M}
   \underbrace{\Re\{[\tilde{\mathbf{w}}_{k}]_{m}\}\,\dot{b}(\hat{\theta}_{k} z_{m}^{(n)})
   + \Im\{[\tilde{\mathbf{w}}_{k}]_{m}\}\,\ddot{b}(\hat{\theta}_{k} z_{m}^{(n)})}_{\mathclap{\text{constant}}}, \label{g2k_LB}
\end{align}
\normalsize
where Hessian matrix of the quadratic term is denoted as
\begin{equation} \label{g2k_LB_Hessian}
    {{\mathbf{R}}_{g2}}=-2\hat{\theta }_{k}^{2}\text{diag}\left( \Re \left\{ {{\left[ {{{\mathbf{\tilde{w}}}}_{k}} \right]}_{m}} \right\}+\Im \left\{ {{\left[ {{{\mathbf{\tilde{w}}}}_{k}} \right]}_{m}} \right\} \right)_{m=1}^{M}.
\end{equation}
The linear part of (\ref{g2k_LB}) can be re-written as
\small
\begin{equation} \label{g2k_LB_linear}
    {{\left[ {{\mathbf{r}}_{g2}} \right]}_{p}}=2{{\hat{\theta }}_{k}}\left( \Re \left\{ {{\left[ {{{\mathbf{\tilde{w}}}}_{k}} \right]}_{p}} \right\}\dot{a}\left( {{{\hat{\theta }}}_{k}}z_{p}^{\left( n \right)} \right)+\Im \left\{ {{\left[ {{{\mathbf{\tilde{w}}}}_{k}} \right]}_{p}} \right\}\ddot{a}\left( {{{\hat{\theta }}}_{k}}z_{p}^{\left( n \right)} \right) \right).
\end{equation}
\normalsize

Based on (\ref{g1kcos_LB_Hessian}), (\ref{g2k_LB_Hessian}), (\ref{g1kcos_LB_linear}), and (\ref{g2k_LB_linear}), the Hessian and linear coefficients can be combined as 
\begin{align} 
    {{\mathbf{R}}_{k}}=&{{\left| {{{\bar{\ell }}}_{k}} \right|}^{2}}\mathbf{R}_{g1}^{\cos }+{{\mathbf{R}}_{g2}}, \label{final_Sk_Hessian} \\ 
    {{\mathbf{r}}_{k}}=&{{\left| {{{\bar{\ell }}}_{k}} \right|}^{2}}\mathbf{r}_{g1}^{\cos }+{{\mathbf{r}}_{g2}}, \label{final_Sk_linear}  
\end{align}
respectively. Therefore, the quadratic expression (\ref{FA_comm_useful_V1}) can be lower-bounded as in (\ref{FA_comm_useful_LB}), where the corresponding constant term $c_k$ is written as (\ref{final_Sk_const}), at the top of next page.

\begin{floatEq}
	\begin{subequations}\label{const_terms_interf_comm}\begin{align} 
		c_k = & \left|\bar{\ell}_k\right|^{2}
        \sum_{m=1}^{M}\sum_{m'=1}^{M} {{\varpi }_{k,m,{m}'}}\left( \dot{b}\left( \varphi _{k,m,{m}'}^{\left( n \right)} \right)-\dot{a}\left( \varphi _{k,m,{m}'}^{\left( n \right)} \right)\Delta {{\psi }_{k,m,{m}'}}-\frac{1}{2}\Delta \psi _{k,m,{m}'}^{2} \right)   \notag\\
        & + 2 \sum_{m=1}^{M}
        \Re \left\{ {{\left[ {{{\mathbf{\tilde{w}}}}_{k}} \right]}_{m}} \right\}\dot{b}\left( {{{\hat{\theta }}}_{k}}z_{m}^{\left( n \right)} \right)+\Im \left\{ {{\left[ {{{\mathbf{\tilde{w}}}}_{k}} \right]}_{m}} \right\}\ddot{b}\left( {{{\hat{\theta }}}_{k}}z_{m}^{\left( n \right)} \right) +{{\left| {{{\tilde{\ell }}}_{k}}\mathbf{w}_{k}^{H}{{{\mathbf{\tilde{h}}}}_{k}} \right|}^{2}}, \label{final_Sk_const} \\       
		  {{c}_{k,i}} = & {{\left| {{{\bar{\ell }}}_{k}} \right|}^{2}}\sum\limits_{m}^{M}{\sum\limits_{{{m}'}}^{M}{{{\varpi }_{i,m,{m}'}}\left( \dot{d}\left( \varphi _{k,i,m,{m}'}^{\left( n \right)} \right)-\dot{c}\left( \varphi _{k,i,m,{m}'}^{\left( n \right)} \right)\Delta {{\psi }_{i,m,{m}'}}+\frac{1}{2}\Delta \psi _{i,m,{m}'}^{2} \right)}} \notag \\ 
        & +2\sum\limits_{m=1}^{M}{\Re \left\{ {{\left[ {{{\mathbf{\tilde{w}}}}_{k,i}} \right]}_{m}} \right\}\dot{d}\left( {{{\hat{\theta }}}_{k}}z_{m}^{\left( n \right)} \right)+\Im \left\{ {{\left[ {{{\mathbf{\tilde{w}}}}_{k,i}} \right]}_{m}} \right\}\ddot{d}\left( {{{\hat{\theta }}}_{k}}z_{m}^{\left( n \right)} \right)}+{{\left| {{{\tilde{\ell }}}_{k}}\mathbf{w}_{i}^{H}{{{\mathbf{\tilde{h}}}}_{k}} \right|}^{2}},  \label{const_comm_interf_term1_} \\
        {{c}_{j,k,t}} = & {{\left| {{{\bar{\omega }}}_{k,t}} \right|}^{2}}\sum\limits_{m}^{M}{\sum\limits_{{{m}'}}^{M}{{{\varpi }_{j,m,{m}'}}\left( \dot{d}\left( \varphi _{j,t,m,{m}'}^{\left( n \right)} \right)-\dot{c}\left( \varphi _{j,t,m,{m}'}^{\left( n \right)} \right)\Delta {{\psi }_{j,m,{m}'}}+\frac{1}{2}\Delta \psi _{j,m,{m}'}^{2} \right)}} \notag \\ 
        & +2\sum\limits_{m=1}^{M}{\Re \left\{ {{\left[ {{{\mathbf{\tilde{w}}}}_{j,k,t}} \right]}_{m}} \right\}\dot{d}\left( {{{\hat{\theta }}}_{t}}z_{m}^{\left( n \right)} \right)+\Im \left\{ {{\left[ {{{\mathbf{\tilde{w}}}}_{j,k,t}} \right]}_{m}} \right\}\ddot{d}\left( {{{\hat{\theta }}}_{t}}z_{m}^{\left( n \right)} \right)}+{{\left| {{{\tilde{\omega }}}_{k,t}}\mathbf{w}_{j}^{H}{{{\mathbf{\tilde{h}}}}_{t}} \right|}^{2}}. \label{const_comm_interf_term4_} 
	\end{align}\end{subequations}
\end{floatEq}

\section{Communication interference terms (\ref{FA_UB_comm_interf})} \label{appendix2}
To compact the expressions of the interference $\mathcal{I}_k^{\mathrm{UB}}(\mathbf{z})$, in  (\ref{FA_UB_comm_interf}), we define the diagonal operator
\begin{equation}\label{def_Dop}
\mathcal{D}(\mathbf{x}) \triangleq \mathrm{diag}\!\Big(\Re\{[\mathbf{x}]_{m}\}+\Im\{[\mathbf{x}]_{m}\}\Big)_{m=1}^{M}.
\end{equation}
By employing the analysis of Appendix \ref{appendix1} and upper-bounds of (\ref{UB_sin_cos}) on trigonometric functions, the Hessian-related interference matrices can be denoted as
\begin{align*}
\mathbf{\bar R}_{k}
&=2\hat{\theta}_{k}^{2}\!\sum_{\substack{i=1, i\neq k}}^{K}\!|\bar{\ell}_{k}|^{2}\mathbf{\Omega}_{i}
+\mathcal{D}\!\big(\mathbf{\tilde w}_{k,i}\big), \\
\mathbf{\tilde R}_{k}
&=2\hat{\theta}_{k}^{2}\!\sum_{r=1}^{M}\lambda_{r}\!\left(|\bar{\ell}_{k}|^{2}\mathbf{\Lambda}_{r}
+\mathcal{D}\!\big(\mathbf{\tilde e}_{k,r}\big)\right), \\
\bm{\overset{\scriptscriptstyle\smile}{\mathbf R}}_{k}
&=2\!\sum_{t=1}^{T}\!\beta_{t}\hat{\theta}_{t}^{2}\!\sum_{r=1}^{M}\lambda_{r}\!\left(|\bar{\omega}_{k,t}|^{2}\mathbf{\Lambda}_{r}
+\mathcal{D}\!\big(\mathbf{\tilde e}_{k,t,r}\big)\right), \\
\bm{\overset{\scriptscriptstyle\frown}{\mathbf R}}_{k}
&=2\!\sum_{t=1}^{T}\!\beta_{t}\hat{\theta}_{t}^{2}\!\sum_{j=1}^{K}\!\left(|\bar{\omega}_{k,t}|^{2}\mathbf{\Omega}_{j}
+\mathcal{D}\!\big(\mathbf{\tilde w}_{j,k,t}\big)\right). 
\end{align*}
The corresponding linear terms can be expressed written as
\begin{align*}
\mathbf{\bar r}_{k}
&=\sum_{\substack{i=1, i\neq k}}^{K}\!\Big(|\bar{\ell}_{k}|^{2}\mathbf{r}_{g_{1,k}}^{\cos}(\mathbf{w}_{i})
+\mathbf{r}_{g_{2,k}}(\mathbf{w}_{i})\Big), \\
\mathbf{\tilde r}_{k}
&=\sum_{r=1}^{M}\lambda_{r}\Big(|\bar{\ell}_{k}|^{2}\mathbf{r}_{g_{1,k,r}}^{\cos}(\mathbf{e}_{r})
+\mathbf{r}_{g_{2,k,r}}(\mathbf{e}_{r})\Big), \\
\bm{\overset{\scriptscriptstyle\smile}{\mathbf r}}_{k}
&=\sum_{t=1}^{T}\beta_{t}\sum_{r=1}^{M}\lambda_{r}\Big(|\bar{\omega}_{k,t}|^{2}\mathbf{r}_{g_{1,t,r}}^{\cos}(\mathbf{e}_{r})
+\mathbf{r}_{g_{2,k,t,r}}(\mathbf{e}_{r})\Big), \\ 
\bm{\overset{\scriptscriptstyle\frown}{\mathbf r}}_{k}
&=\sum_{t=1}^{T}\beta_{t}\sum_{j=1}^{K}\Big(|\bar{\omega}_{k,t}|^{2}\mathbf{r}_{g_{1,j,t}}^{\cos}(\mathbf{w}_{j})
+\mathbf{r}_{g_{2,j,k,t}}(\mathbf{w}_{j})\Big). 
\end{align*}
In addition, the constant terms are given by
\begin{align*}
    {{\bar{c}}_{k}}=&\sum\limits_{i=1,i\ne k}^{K}{{{c}_{k,i}}}, \quad 
    {{\tilde{c}}_{k}}=\sum\limits_{r=1}^{M}{{{\lambda }_{r}}{{c}_{k,r}}}, \notag \\
    {{\overset{\scriptscriptstyle\smile}{c}}_{k}}=&\sum\limits_{t=1}^{T}{{{\beta }_{t}}\sum\limits_{r=1}^{M}{{{\lambda }_{r}}{{c}_{k,t,r}}}}, \quad {{\overset{\scriptscriptstyle\frown}{c}}_{k}}=\sum\limits_{t=1}^{T}{{{\beta }_{t}}\sum\limits_{j=1}^{K}{{{c}_{j,k,t}}}},
\end{align*}
where $c_{k,i}$ and $c_{j,k,t}$ are defined in (\ref{const_comm_interf_term1_}) and  (\ref{const_comm_interf_term4_}),  respectively, at the top of next page. 
Note that $c_{k,r}$ can be obtained from $c_{k,i}$ by replacing
$\varpi_{i,m,m'}$ with $\nu_{r,m,m'}$, 
$\varphi_{k,i,m,m'}^{(n)}$ with $\varphi_{r,k,m,m'}^{(n)}$,
$\Delta\psi_{i,m,m'}$ with $\Delta\psi_{r,m,m'}$,
$\tilde{\mathbf{w}}_{k,i}$ with $\tilde{\mathbf{e}}_{k,r}$,
and $\mathbf{w}_{i}$ with $\mathbf{e}_{r}$. Similarly, $c_{k,t,r}$ can be obtained from $c_{j,k,t}$ by replacing
$\varpi_{j,m,m'}$ with $\nu_{r,m,m'}$,
$\varphi_{j,t,m,m'}^{(n)}$ with $\varphi_{r,t,m,m'}^{(n)}$,
$\Delta\psi_{j,m,m'}$ with $\Delta\psi_{r,m,m'}$,
$\tilde{\mathbf{w}}_{j,k,t}$ with $\tilde{\mathbf{e}}_{k,t,r}$,
and $\mathbf{w}_{j}$ with $\mathbf{e}_{r}$.
Moreover, we have  
${{\mathbf{\Omega }}_{i}}=\left( \sum\limits_{m}^{M}{\left| {{\left[ {{\mathbf{w}}_{i}} \right]}_{m}} \right|} \right)\text{diag}\left( \left| {{\pmb{\omega }}_{i}} \right| \right)-\left| {{\pmb{\omega }}_{i}} \right|{{\left| {{\pmb{\omega }}_{i}} \right|}^{T}}$,
${{\mathbf{\Lambda }}_{r}}=\left( \sum\limits_{m}^{M}{\left| {{\left[ {{\mathbf{e}}_{r}} \right]}_{m}} \right|} \right)\text{diag}\left( \left| {{\pmb{\nu }}_{r}} \right| \right)-\left| {{\pmb{\nu }}_{r}} \right|{{\left| {{\pmb{\nu }}_{r}} \right|}^{T}}$,  
${{\mathbf{\tilde{w}}}_{k,i}}={{\bar{\ell }}_{k}}\tilde{\ell }_{k}^{*}\mathbf{\tilde{h}}_{k}^{H}{{\mathbf{W}}_{i}}$,  
${{\varphi }_{k,i,m,{m}'}}={{\hat{\theta }}_{k}}{{z}_{m}}-{{\hat{\theta }}_{k}}{{z}_{{{m}'}}}-\Delta {{\psi }_{i,m,{m}'}}$,  
$\left| {{\pmb{\nu }}_{r}} \right|={{\left[ \left| {{\left[ {{\mathbf{e}}_{r}} \right]}_{1}} \right|,...,\left| {{\left[ {{\mathbf{e}}_{r}} \right]}_{M}} \right| \right]}^{T}}$, and 
\small
\begin{align*} 
    {{\left[ \mathbf{r}_{{{g}_{1,k}}}^{\cos } \right]}_{p}}=&{{{\hat{\theta }}}_{k}}\sum\limits_{{{m}'}}^{M}{{{\varpi }_{i,p,{m}'}}\left( \dot{c}\left( \varphi _{k,i,p,{m}'}^{\left( n \right)} \right)-\Delta {{\psi }_{i,p,{m}'}} \right)} \\ 
    & -{{{\hat{\theta }}}_{k}}\sum\limits_{{{m}'}}^{M}{{{\varpi }_{i,{m}',p}}\left( \dot{c}\left( \varphi _{k,i,{m}',p}^{\left( n \right)} \right)-\Delta {{\psi }_{i,{m}',p}} \right)},\\
    {{\left[ {{\mathbf{r}}_{{{g}_{2,k}}}} \right]}_{p}}=&2{{\hat{\theta }}_{k}}\left( \Re \left\{ {{\left[ {{{\mathbf{\tilde{w}}}}_{k,i}} \right]}_{p}} \right\}\dot{c}\left( {{{\hat{\theta }}}_{k}}z_{p}^{\left( n \right)} \right)+\Im \left\{ {{\left[ {{{\mathbf{\tilde{w}}}}_{k,i}} \right]}_{p}} \right\}\ddot{c}\left( {{{\hat{\theta }}}_{k}}z_{p}^{\left( n \right)} \right) \right), \\
    {{\left[ \mathbf{r}_{{{g}_{1,k,r}}}^{\cos } \right]}_{p}}=&{{{\hat{\theta }}}_{k}}\sum\limits_{{{m}'}}^{M}{{{\nu }_{r,p,{m}'}}\left( \dot{c}\left( \varphi _{r,k,p,{m}'}^{\left( n \right)} \right)-\Delta {{\psi }_{r,p,{m}'}} \right)} \\ 
    &  -{{{\hat{\theta }}}_{k}}\sum\limits_{{{m}'}}^{M}{{{\nu }_{r,{m}',p}}\left( \dot{c}\left( \varphi _{r,k,{m}',p}^{\left( n \right)} \right)-\Delta {{\psi }_{r,{m}',p}} \right)}, \\
    {{\left[ {{\mathbf{r}}_{{{g}_{2,k,r}}}} \right]}_{p}}=&2{{\hat{\theta }}_{k}}\left( \Re \left\{ {{\left[ {{{\mathbf{\tilde{e}}}}_{k,r}} \right]}_{p}} \right\}\dot{c}\left( {{{\hat{\theta }}}_{k}}z_{p}^{\left( n \right)} \right)+\Im \left\{ {{\left[ {{{\mathbf{\tilde{e}}}}_{k,r}} \right]}_{p}} \right\}\ddot{c}\left( {{{\hat{\theta }}}_{k}}z_{p}^{\left( n \right)} \right) \right),
\end{align*}
\normalsize
where the entries of $\mathbf{r}_{g_{1,t,r}}^{\cos}$ can be obtained from those of 
$\mathbf{r}_{g_{1,k,r}}^{\cos}$ by replacing $\hat{\theta}_{k}$ with $\hat{\theta}_{t}$ and
$\varphi_{r,k,p,m'}^{(n)}$ and $\varphi_{r,k,m',p}^{(n)}$ with 
$\varphi_{r,t,p,m'}^{(n)}$ and $\varphi_{r,t,m',p}^{(n)}$, respectively.
Furthermore, $\mathbf{r}_{g_{1,j,t}}^{\cos}$ follows from $\mathbf{r}_{g_{1,t,r}}^{\cos}$ by replacing
$\nu_{r,p,m'}$ and $\nu_{r,m',p}$ with $\varpi_{j,p,m'}$ and $\varpi_{j,m',p}$, and
$\Delta\psi_{r,p,m'}$ and $\Delta\psi_{r,m',p}$ with 
$\Delta\psi_{j,p,m'}$ and $\Delta\psi_{j,m',p}$, respectively. Similarly, the entries of $\mathbf{r}_{g_{2,k,t,r}}$ can be obtained from those of 
$\mathbf{r}_{g_{2,k,r}}$ by replacing $\hat{\theta}_{k}$ with $\hat{\theta}_{t}$ and
$\tilde{\mathbf{e}}_{k,r}$ with $\tilde{\mathbf{e}}_{k,t,r}$.
In turn, $\mathbf{r}_{g_{2,j,k,t}}$ follows from $\mathbf{r}_{g_{2,k,t,r}}$ by replacing
$\tilde{\mathbf{e}}_{k,t,r}$ with $\tilde{\mathbf{w}}_{j,k,t}$.

\bibliographystyle{ieeetr}
\bibliography{Refs}

\end{document}